\documentclass[12pt]{article}
\usepackage{amsmath}
\usepackage{amsthm,amssymb}
\usepackage{graphicx}
\usepackage{enumerate}
\usepackage{natbib}
\usepackage{url} 
\usepackage{epsfig,multirow}
\usepackage{mathtools}
\usepackage{setspace}
\usepackage{paralist}
\usepackage{algorithm}
\usepackage{appendix}
\usepackage[noend]{algorithmic}
\usepackage{bm,verbatim,subfigure}
\usepackage{color}
\usepackage{booktabs}

\usepackage[displaymath]{lineno}

\usepackage{etoolbox} 
\makeatletter 
\newcommand*\linenomathpatch{\@ifstar{\linenomathpatch@AMS}{\linenomathpatch@}}
\newcommand*\linenomathpatch@[1]{
  \expandafter\pretocmd\csname #1\endcsname {\linenomathWithnumbers}{}{}
  \expandafter\pretocmd\csname #1*\endcsname{\linenomathWithnumbers}{}{}
  \expandafter\apptocmd\csname end#1\endcsname {\endlinenomath}{}{}
  \expandafter\apptocmd\csname end#1*\endcsname{\endlinenomath}{}{}
}
\newcommand*\linenomathpatch@AMS[1]{
  \expandafter\pretocmd\csname #1\endcsname {\linenomathWithnumbersAMS}{}{}
  \expandafter\pretocmd\csname #1*\endcsname{\linenomathWithnumbersAMS}{}{}
  \expandafter\apptocmd\csname end#1\endcsname {\endlinenomath}{}{}
  \expandafter\apptocmd\csname end#1*\endcsname{\endlinenomath}{}{}
}
\let\linenomathWithnumbersAMS\linenomathWithnumbers
\patchcmd\linenomathWithnumbersAMS{\advance\postdisplaypenalty\linenopenalty}{}{}{}
\makeatother 

\linenomathpatch{equation}
\linenomathpatch*{gather}
\linenomathpatch*{multline}
\linenomathpatch*{align}
\linenomathpatch*{alignat}
\linenomathpatch*{flalign}

\usepackage{soul}

\newcommand{\blind}{1}

\usepackage[left=1in,top=1in,right=1in,bottom=1in]{geometry}

\newtheorem{proposition}{Proposition}
\newtheorem{theorem}{Theorem}

\newtheorem{definition}{Definition}[section]
\newtheorem{condition}{Condition}

\bibpunct{(}{)}{;}{a}{,}{,}

\usepackage{lscape}

\makeatletter
\newcommand*{\addFileDependency}[1]{
  \typeout{(#1)}
  \@addtofilelist{#1}
  \IfFileExists{#1}{}{\typeout{No file #1.}}
}
\makeatother

\usepackage[colorlinks=true,linkcolor=blue,citecolor=blue,filecolor=blue]{hyperref}

\def\bg{\begin{figure}[tpbh]\begin{center}}
\def\eg{\end{center}\end{figure}}

\long\def\symbolfootnote[#1]#2{\begingroup\def\thefootnote{\fnsymbol{footnote}}
\footnote[#1]{#2}\endgroup}

\newcommand{\Lower}[1]{\smash{\lower 1.5ex\hbox{#1}}}

\newcommand{\beqn}{\begin{eqnarray}}
\newcommand{\eeqn}{\end{eqnarray}}
\newcommand{\beq}{\begin{equation}}
\newcommand{\eeq}{\end{equation}}
\def\beqs{\begin{equation*}}
\def\eeqs{\end{equation*}}
\def\.{$.$}

\def\bsigma{{\mbox{\boldmath $\sigma$}}}
\def\bSigma{{\mbox{\boldmath $\Sigma$}}}

\def\bpi{{\mbox{\boldmath $\pi$}}}
\def\bmu{{\mbox{\boldmath $\mu$}}}
\def\btheta{{\mbox{\boldmath $\theta$}}}
\def\bbeta{{\mbox{\boldmath $\beta$}}}

\newcommand{\bb}{\mbox{\bf b}}

\newcommand{\be}{\mbox{\bf e}}

\newcommand{\bu}{\mbox{\bf u}}

\newcommand{\bx}{\mbox{\bf x}}
\newcommand{\by}{\mbox{\bf y}}
\newcommand{\bz}{\mbox{\bf z}}

\newcommand{\bD}{\mbox{\bf D}}

\newcommand{\bI}{\mbox{\bf I}}

\newcommand{\bV}{\mbox{\bf V}}

\newcommand{\bX}{\mbox{\bf X}}

\newcommand{\B}{\mathbf}

\def\trans{^{\rm T}}
\def\strans{^{\rm *T}}

\def\0{{\bf 0}}
\def\A{{\bf A}}

\def\I{{\bf I}}

\def\T{{\bf T}}

\def\W{{\bf W}}

\def\X{{\bf X}}
\def\x{{\bf x}}
\def\T{{\bf T}}

\def\1{{\bf 1}}

\def\J{{\bf J}}

\def\brho{{\boldsymbol \rho}}

\def\bphi{{\boldsymbol \phi}}
\def\bPhi{{\boldsymbol \Phi}}
\def\bg{{\boldsymbol \gamma}}

\def\bpi{{\boldsymbol \pi}}
\def\btheta{{\boldsymbol \theta}}
\def\bvartheta{{\boldsymbol \vartheta}}

\def\bSigma{{\bf \Sigma}}

\def\bmu{{\boldsymbol \mu}}
\def\calA{\mathcal A}

\def\calT{\mathcal T}

\def\trans{^{\rm T}}

\def\be{\begin{eqnarray}}
\def\ee{\end{eqnarray}}
\def\bse{\begin{eqnarray*}}
\def\ese{\end{eqnarray*}}

\def\red{\color{red}}

\def\boxit#1{\vbox{\hrule\hbox{\vrule\kern6pt\vbox{\kern6pt#1\kern6pt}\kern6pt\vrule}\hrule}}

\def\revcolor#1{}

\makeatletter
\def\@fnsymbol#1{\ensuremath{\ifcase#1\or *\or **\or \dagger\or \ddagger\or
   \mathsection\or \mathparagraph\or \|\or \dagger\dagger
   \or \ddagger\ddagger \else\@ctrerr\fi}}
\makeatother

\newcommand*{\email}[1]{%
    \href{mailto:#1}{\color{black}{#1}}\par
    }

\newcommand{\mathst}[1]
{\bgroup\mathchoice
  {\sbox0{$\displaystyle{#1}$}%
    \usebox0\hspace{-\wd0}%
    \rule[0.5\ht0-0.5\dp0-.5pt]{\wd0}{1pt}}%
  {\sbox0{$\textstyle{#1}$}%
    \usebox0\hspace{-\wd0}%
    \rule[0.5\ht0-0.5\dp0-.5pt]{\wd0}{1pt}}%
  {\sbox0{$\scriptstyle{#1}$}%
    \usebox0\hspace{-\wd0}%
    \rule[0.5\ht0-0.5\dp0-.5pt]{\wd0}{1pt}}%
  {\sbox0{$\scriptscriptstyle{#1}$}%
    \usebox0\hspace{-\wd0}%
    \rule[0.5\ht0-0.5\dp0-.5pt]{\wd0}{1pt}}%
  \egroup}
\newcommand{\calP}{\mathcal{P}_{\gamma}}
\newcommand{\hide}[1]{}
\renewcommand{\mathst}[1]{}
\renewcommand{\red}{}

\begin{document}



\date{\vspace{-5ex}}

\if1\blind
{
  \title{Pursuing Sources of Heterogeneity in Modeling Clustered Population}
  \author{{Yan Li$^{1,}$},
    {Chun Yu$^{2,}$},
    {Yize Zhao$^3$},
    {Weixin Yao$^4$},\\
    {Robert H. Aseltine$^5$},
    {and Kun Chen$^{1,5,}$}\thanks{Corresponding author. Email: \email{kun.chen@uconn.edu}}\\\\
    $^1$Department of Statistics, University of Connecticut, Storrs, CT\\
    $^2$School of Statistics, Jiangxi University of Finance and Economics, China\\
    $^3$Department of Biostatistics, Yale School of Public Health, New Haven, CT\\
    $^4$Department of Statistics, University of California, Riverside,CA\\
    $^5$Center for Population Health, University of Connecticut Health Center, \\Farmington, CT
  }
  \maketitle
} \fi

\if0\blind
{
  \bigskip
  \bigskip
  \bigskip
  \begin{center}
    {\LARGE Pursuing Sources of Heterogeneity in Modeling Clustered Population}
  \end{center}
  \medskip
} \fi

\begin{abstract}
  Researchers often have to deal with heterogeneous population with mixed regression relationships, increasingly so in the era of data explosion. In such problems, when there are many candidate predictors, it is not only of interest to identify the predictors that are associated with the outcome, but also to distinguish the true \textit{sources of heterogeneity}, i.e., to identify the predictors that have different effects among the clusters and thus are the true contributors to the formation of the clusters. We clarify the concepts of the source of heterogeneity that account for potential scale differences of the clusters and propose a \textit{regularized finite mixture effects regression} to achieve heterogeneity pursuit and feature selection simultaneously. We develop an efficient algorithm and show that our approach can achieve both estimation and selection consistency. Simulation studies further demonstrate the effectiveness of our method under various practical scenarios. Three applications are presented, namely, an imaging genetics study for linking genetic factors and brain neuroimaging traits in Alzheimer's disease, a public health study for exploring the association between suicide risk among adolescents and their school district characteristics, and a sport analytics study for understanding how the salary levels of baseball players are associated with their performance and contractual status.
\end{abstract}
\vfill

\numberwithin{equation}{section} 

\newpage
\doublespacing
\section{Introduction}

Regression is a fundamental statistical problem, of which a prototype is to model a response $y\in \mathbb{R}$ as a function of a $p$-dimensional predictor vector $\bx$. In many applications, the classical assumption that the conditional association between $y$ and $\x$ is homogeneous in the population does not hold. Rather, their conditional association may vary across several latent sub-populations or clusters. Such population heterogeneity can be modeled by a finite mixture regression (FMR), 
which is capable of identifying the clusters by learning multiple models together. Since first introduced by \citet{goldfeld73}, FMR has been further developed in various directions and is widely used in various fields; see, e.g., \citet{jiang99}, \citet{Bohning99},  \citet{mclachlan2004finite}, and \citet{ChenMishra2017JASA}.

In the era of data explosion, regression problems with a large sample size and/or a large number of variables become increasingly common, which makes the modeling of population heterogeneity even more relevant. However, while many high-dimensional methods have been developed for mixture regression \citep{khalili2012variable,stadler2010,khalili2011overview}, utilizing regularization has been mainly for the purpose of variable selection, i.e., to identify the predictors that are relevant to the modeling of the outcome.

In this paper, we tackle a challenging and interesting problem in the context of mixture model: to identify the predictors that are truly the sources of heterogeneity. That is, besides the selection of important predictors, we aim to further divide the selected predictors into two categories, the ones that only have common effects on the outcome and the ones that have different effects in different clusters. {\revcolor{red} Being able to identify the sources of heterogeneity not only could reduce the complexity of the mixture model, but also could improve the model interpretability and enable us to gain deeper insights on the outcome-predictor association.}

One important field that motivates our study is the imaging genetics with application to mental disorders such as Alzheimer's disease. As demonstrated by twin studies \citep{cauwenberghe2016genetmed}, genetic factors play an import role in Alzheimer's disease and offers great promise for disease modeling and drug development. Compared with categorical diagnoses, neuroimaging trait has distinct advantages to capture disease etiology, and has been used in replacement of conventional clinical behavioral phenotypes in genome wide association studies (GWAS). Due to the availability of large-scale brain imaging and genetics data in landmark studies like the Alzheimer's Disease Neuroimaging Initiative \citep{weiner2013aad}, a large body of literature in  imaging genetics focuses on high-dimensional modeling to identify risk genetic variants \citep{vounou2012neuroimage,lu2015genetepidemiol,zhao2019structured}. However, a major challenge in the field that has not been well investigated is how to link the imaging-associated genetic factors to actual disease diagnosis or progression and provide meaningful interpretations. Specifically, for progressive mental illness like Alzheimer's disease, it is critical to identify biomarkers that can predict the disease at early time. Therefore, we believe that not only there are genetic factors that impact overall disease risk, but also there are the ones that have differential impacts across some sub-groups which may be corresponding to different progressive periods/stages. While a few attempts have been made to bridge the pathological paths among genotype, imaging and clinical outcomes \citep{hao2017mining,bi2017genome,xu2017imaging}, to the best of our knowledge, none of the existing methods consider the heterogeneity within patient cohort or imaging endophenotype, nor are they capable to identify genetic factors that give arise disease sub-groups.

Indeed, the problem of heterogeneity pursuit is prevelent in various fields, ranging from genetics, population health, to even sports analytics. In a study on suicide risk among adolescents, we used  data  from  the  State of Connecticut to explore the association between suicide risk among 15-19 year old and the characteristics of their school districts. It is of great interest to learn whether different association patterns co-exist and whether they are due to the differences in demographic, social-economic, and/or academic factors of the school districts. In a study on major league baseball players, the goal is to find out which performance measures and contract/free agent statues of the players contributed to the formation of distinct salary mechanisms or clusters.

In this work, we propose a \textit{regularized finite mixture effects regression} model to perform feature selection and identify sources of heterogeneity simultaneously. The problem is formulated using the effects model parameterization (in analogous to the formulations used in analysis of variance), that is, the effect of each predictor on the outcome is decomposed to a common effect term and a set of cluster-specific terms that are constrained to sum up to zero. We consider adaptive $\ell_1$ penalization on both the cluster-specific effect parameters and common effect parameters, which leads to the identification of the relevant variables and those with heterogeneous effects. The model estimation is conducted via an Expectation-Maximization (EM) algorithm, in which the M step results in a linearly constrained $\ell_1$ penalized regression and is solved by a Bregman coordinate descent algorithm \citep{Bregman1967,Goldstein2009}. We show that the proposed approach can also be cast as a regularized finite mixture regression with a generalized lasso penalty; this connection facilitates our theoretical analysis in showing the estimation and selection consistency. Although we mainly focus on normal mixture model and $\ell_1$ regularization, our approach can be readily generalized to other non-Gaussian models with broad class of penalties and constraints. A user-friendly R package is developed for practitioners to apply our approach.


\vspace{-1em}
\section{Mixture Effects Model For Heterogeneity Pursuit}\label{sec:model}
\subsection{An Overview of Finite Mixture Regression (FMR)}

We start with a description of the classical normal finite mixture regression (FMR). Let $y \in \mathbb{R}$ be a response/outcome variable and $\x = (x_1,\ldots,x_p)\trans \in \mathbb{R}^{p}$ be a $p$-dimensional predictor vector. In FMR with $m$ components, it is assumed that a linear regression model holds for each of the $m$ components, i.e., with probability $\pi_j$, a random sample $(y,\bx)$ belongs to the $j$th mixture component ($j=1,\ldots,m$), for which we have that $y=\bx\trans \bb_j+\epsilon_{j}$, where $\bb_j \in \mathbb{R}^p$ is a fixed and unknown coefficient vector, and $\epsilon_{j}\sim N(0,\sigma_j^2)$ with $\sigma_j^2>0$. For the ease of notation, here the intercept term is included by setting the first element of $\bx$ as one. Therefore, the conditional probability density function of $y$ given $\bx$ is
\begin{align}
\sum_{j=1}^{m}\pi_{j}\frac{1}{\sqrt{2\pi}\sigma_j}\exp\{\frac{-(y-\bx\trans\bb_{j})^2}{2\sigma_j^2}\}
\label{eq:mixreg1}
\end{align}
where $\pi_j$'s are the mixing probabilities satisfying $\pi_j > 0$, $\sum_{j=1}^{m}\pi_{j} = 1$. We write $(\bb_1,\ldots,\bb_m) = (\widetilde{\bb}_1,\ldots,\widetilde{\bb}_p)\trans$, where $\widetilde{\bb}_k \in \mathbb{R}^{m}$ collects $m$ component-specific coefficients for the predictor $x_k$. With finite samples, the maximum likelihood approach is often used for parameter estimation and inference in FMR, via the celebrated EM algorithm \citep{dempster1977maximum} and its many variates \citep{meng1991jasa, meng1993biometrika}. \citet{khalili2012variable} was among the first to propose penalized likelihood approach for variable selection in FMR models; asymptotic properties were established in their work under the fixed $p$, large $n$ paradigm. \citet{stadler2010} studied $\ell_1$ penalized FMR and derived estimation errors bounds and selection consistency under general high-dimensional setups. \citet{Khalili2013} further studied penalized FMR for a general family of penalty functions. Other relevant works include \citet{wedel1995mixture}, \citet{weruaga2015sparse}, \citet{bai2016mixture}, and \citet{dougru2016parameter}. For a comprehensive review, see, e.g., \citet{khalili2011overview}. The penalized FMR models have been widely applied in many real-world problems, such as gene expression analysis \citep{Xie2008}, disease progression subtyping \citep{Gao2016}, multi-species distribution modeling \citep{Hui15}, protein clustering \citep{ChenMishra2017JASA}, among others.

In the above mixture setup, the variance parameters $\sigma_j^2$ play important roles. Unlike in regular linear regression where its single variance parameter generally can be treated as nuisance in the estimation of the regression coefficients, the variance parameters in mixture models directly impact on the scaling (thus interpretation) and estimation of the regression coefficients of the multiple mixture components, and consequently, they also affect the assessment and even the definition of ``heterogeneous regression effects''. To facilitate the further discussion, we present a re-scaled version of FMR \citep{stadler2010}, 
\begin{align*}
\bphi_j = \frac{\bb_j}{\sigma_j} = (\phi_{1j},\ldots, \phi_{pj})\trans,\,\, \rho_j = \sigma_j^{-1} 
\quad (j = 1,\ldots, m),
\end{align*}
and subsequently rewrite the conditional density in \eqref{eq:mixreg1} as 
\begin{align}
f(y\mid \bx,\bvartheta)= \sum_{j=1}^{m}\pi_{j}\frac{\rho_j}{\sqrt{2\pi}}\exp\{-\frac{1}{2}(\rho_jy-\bx\trans\bphi_{j})^2\}\label{eq:mixreg2},
\end{align}
where $\bvartheta = (\bphi_1, \ldots,\bphi_m;\pi_1,\ldots,\pi_m;\rho_1,\ldots,\rho_m)$ collects all the unknown parameters. We write
\[
\bPhi = (\bphi_1,\ldots,\bphi_m) = (\widetilde{\bphi}_1,\ldots,\widetilde{\bphi}_p)\trans \in \mathbb{R}^{p\times m}, \qquad \bphi = \mbox{vec}(\bPhi\trans) \in \mathbb{R}^{pm}, 
\]
where $\widetilde{\bphi}_k \in \mathbb{R}^{m}$ collects $m$ component-specific regression coefficients for the predictor $x_k$ for $k = 1, \ldots, p$ and $\mbox{vec}(\cdot)$ is the columnwise vectorization operator.

\subsection{Sources of Heterogeneity under Finite Mixture Regression}\label{sec:method:def}

Now let's consider predictor selection and heterogeneity pursuit. A predictor $x_k$ is said to be relevant or important, if $\widetilde{\bb}_{k} \neq \0$, or equivalently, $\widetilde{\bphi}_{k} \neq \0$. Correspondingly, define
\[
\mathcal{S}_R = \{k; 1\leq k \leq p, \widetilde{\bphi}_{k} \neq \0\}
\] 
to be the index set of all the relevant predictors, and let $p_0 = |\mathcal{S}_R|$ denote its size. Estimating $\mathcal{S}_R$ is typically the main task of a variable selection method.

We aim higher. That is, besides identifying the relevant variables, we want to also find out among them which ones actually contribute to the population heterogeneity. {\revcolor{red}However, the concept of ``source of heterogeneity'' is not as easily defined as it appears, since the different mixture components are possibly with different scales. We consider two definitions.}

\begin{definition}\label{def:1}
A predictor $x_k$ is said to be a source of heterogeneity, if $\widetilde{\bb}_{k}\neq c\1$ for any $c\in \mathbb{R}$. 
\end{definition}

\begin{definition}\label{def:2}
A predictor $x_k$ is said to be a scaled source of heterogeneity, if $\widetilde{\bphi}_{k}\neq c\1$ for any $c\in \mathbb{R}$.
\end{definition}

Both definitions have their own merits. Definition \ref{def:1} is in terms of the inequality of each raw coefficient vector $\widetilde{\bb}_k$ appeared in \eqref{eq:mixreg1}, which is simple and aims to draw a direct comparison of the raw effects of $x_k$ in different mixture components regardless of their scales. Definition \ref{def:2} is in terms of the scaled counterpart $\widetilde{\bphi}_k$ in \eqref{eq:mixreg2}, and the motivation is to distinguish the heterogeneity induced by the predictors and that caused by inherit scaling differences. In other words, under the second definition, we compare the standardized effects of $x_k$ in different mixture components after putting them on the same scale. An analogy can be drawn from the familiar analysis of variance context: comparing the means of different groups is mostly appropriate when the groups are with the same variances. Notice that the two definitions become equivalent when the component variances are equal, e.g., $\sigma_1^2 = \cdots = \sigma_m^2$, which is a commonly adopted assumption in mixture regression analysis.

In this work, we shall mainly focus on Definition \ref{def:2}, although our methodologies can be readily modified to handle the alternative definition. Based on Definition \ref{def:2}, let $\mathcal{S}_H = \{k; 1\leq k \leq p, \widetilde{\bphi}_{k} \neq c\1, \forall c\in\mathbb{R}\}$ and $p_{00} = |\mathcal{S}_H|$. {\revcolor{red}Henceforth, our objective is to recover both $\mathcal{S}_R$ and $\mathcal{S}_H$. This can potentially lead to a much more parsimonious and interpretable model.} To see this, consider as above that in a $m$-component mixture model with $p$ predictors, there are $p_0$ relevant variables, and among those, only $p_{00}$ variables are sources of heterogeneity. The classical FMR fits a model with $mp$ free regression parameters, which can be infeasible when $p$ is even moderately large comparing to the sample size. Meanwhile, the best model a sparse predictor selection method can possibly produce would have $mp_0$ free regression parameters. We can do better: since only $p_{00}$ predictors are truly the source of heterogeneity, the best model would have only $p_{0} + (m-1)p_{00}$ regression parameters. The saving can be substantial when $p_{00} \ll p_{0} \ll p$ and/or $m$ is large. As an example, consider one of the simulation settings to be presented in Section~\ref{sec:sim} with $m =3$, $p = 30$, $p_{0}=10$, and $p_{00}= 3$. The classic FMR is with $mp = 90$ regression parameters, the sparse selection method can possibly reduce the number to be $mp_0 = 30$, while our method can possibly further reduce the number to $p_{0} + (m-1)p_{00} = 16$ through identifying the sources of heterogeneity.

\subsection{Regularized Mixture Effects Regression}

Motivated by the so-called effects-model formulation commonly used in analysis of variance models, we propose the following \textit{constrained mixture effects model} formulation, to facilitate the pursuit of the sources of heterogeneity in mixture regression,   
\begin{align}
  f(y\mid \bx,\btheta)= \sum_{j=1}^{m}\pi_{j}\frac{\rho_j}{\sqrt{2\pi}}\exp\{-\frac{1}{2}(\rho_jy-\bx\trans\bbeta_0-\bx\trans\bbeta_{j})^2\}\label{eq:mixreg3},\,
  \mbox{s.t.} \sum_{j=1}^m\beta_{jk}=0, k = 1, \ldots, p,
\end{align}
where $\bbeta_0 = (\beta_{01},\ldots,\beta_{0p})\trans \in \mathbb{R}^{p}$ collects the common effects, and $\bbeta_j=(\beta_{j1},\ldots,\beta_{jp})\trans \in \mathbb{R}^{p}$, $j =1,\ldots, m$, are the coefficient vectors of cluster-specific effects. The equality constraints are necessary to ensure the identifiablility of the parameters. 
We write
\[
\B{B} = (\bbeta_0,\bbeta_1,\ldots,\bbeta_m) = (\widetilde{\bbeta}_1,\ldots,\widetilde{\bbeta}_p)\trans \in \mathbb{R}^{p\times (m+1)}, \qquad \bbeta = \mbox{vec}(\B{B}\trans) \in \mathbb{R}^{p(m+1)},
\]
where $\widetilde{\bbeta}_k=(\beta_{0k},\beta_{1k},\ldots,\beta_{mk})\trans \in \mathbb{R}^{m+1}$ collects the common effect and the $m$ cluster-specific effects for predictor $x_k$. The rest of the terms are similarly defined as in \eqref{eq:mixreg2}, except that we now write 
$
\btheta = (\bbeta_0, \bbeta_1,\ldots,\bbeta_m;\pi_1,\ldots,\pi_m;\rho_1,\ldots,\rho_m)
$ to correct all the parameters under this alternative effects-model parameterization.

Now a predictor $x_k$ is deemed to be relevant whenever $\widetilde{\bbeta}_k \neq \0$. Moreover, a relevant variable is deemed to be a source of heterogeneity only if there exists a $1\leq j\leq m$ such that $\beta_{jk} \neq 0$. As such, variable selection and heterogeneity pursuit can be achieved together through a sparse estimation of $\B{B}$. {\revcolor{red}With $n$ independent samples $\{(y_i, \x_i); i = 1,\ldots, n\}$,} we propose to conduct model estimation by maximizing a constrained penalized log-likelihood criterion, 
\begin{align}
 \max_{\btheta} \left\{\ell^{\gamma}_\lambda(\btheta) \equiv \sum_{i=1}^n\log\left\{ f(y_i\mid \bx_i,\btheta) \right\}-n\lambda\sum_{k=1}^{p}{\red \calP\mathst{\rho}}(\widetilde{\bbeta}_{k})\right\},\,\, 
\mbox{s.t.} \sum_{j=1}^m\beta_{jk}=0, k = 1, \ldots, p,\label{mixreg:commplog-varsel1}
\end{align}
where $f(y\mid\bx,\btheta)$ is the conditional density function from \eqref{eq:mixreg3}, and $\calP(\cdot)$ is a penalty function with $\lambda$ being its tuning parameter; we mainly focus on the $\ell_1$ penalty \citep{tib1996} and its adaptive version \citep{zou2006,huang2008}, i.e., 
\begin{align}
\calP(\widetilde{\bbeta}_{k}) = \sum_{j=0}^{m}w_{jk}|\beta_{jk}|, \qquad w_{jk} = |\widehat{\beta}^0_{j,k}|^{-\gamma} \label{eq:penalty}
\end{align}
where $w_{jk}s$ are the adaptive weights constructed from some initial estimator $\widehat{\beta}^0_{j,k}$,  with $\gamma =0$ corresponding to the non-adaptive version and $\gamma >0$ the adaptive version. Apparently there are many other reasonable choices of penalty functions \citep{fan2001,Khalili2013}, but our choice of $\ell_1$ is simple, convex and yet fundamental for sparse estimation. 

Interestingly, the proposed constrained sparse estimation approach can also be understood as a generalized lasso method \citep{she2010gl,tibshirani2011gl} based on the unconstrained model formulation in \eqref{eq:mixreg2}. To see this, observe that each $\widetilde{\bbeta}_k$ can be written as a function of $\widetilde{\bphi}_k$ as
\begin{align*}
\widetilde{\bbeta}_k = \A\widetilde{\bphi}_k, \qquad \A = 
\begin{pmatrix}
1/m \1_m\trans\\
\I_m - 1/m\J_m
\end{pmatrix}
\in \mathbb{R}^{(m+1)\times m},
\end{align*}
where $\1_m$ is the $m\times 1$ vector of all ones, $\I_m$ is the $m\times m $ identity matrix and $\J_m$ is the $m\times m$ matrix of ones. Therefore, the generalized lasso criterion is expressed as 
\begin{align}
\label{mixreg:commplog-varsel2}
\max_{\bvartheta} & \left\{\l^{\gamma}_\lambda(\bvartheta) \equiv \sum_{i=1}^n\log\left\{ f(y_i\mid \bx_i,\bvartheta) \right\}-n\lambda \|\W(\I_p \otimes \A)\bphi\|_1
\right\},
\end{align}
where $f(y\mid \bx,\bvartheta)$ is the conditional density function from \eqref{eq:mixreg2}, and $\W = \mbox{diag}\{w_{jk}\} \in \mathbb{R}^{p(m+1)\times p(m+1)}$ is constructed from the adaptive weights in \eqref{eq:penalty} accordingly. 
\begin{proposition}
The two problems in \eqref{mixreg:commplog-varsel1} and \eqref{mixreg:commplog-varsel2} are equivalent, in the sense that 
\begin{itemize}
  \item If $\widehat{\bvartheta} = (\widehat{\bphi}, \widehat{\bpi}, \widehat{\brho})$ solves \eqref{mixreg:commplog-varsel2}, then $\widehat{\btheta} = (\widehat{\bbeta}, \widehat{\bpi}, \widehat{\brho})$ solves \eqref{mixreg:commplog-varsel1} where $\widehat{\bbeta} = (\I_p\otimes\A)\widehat{\bphi}$.
  \item And conversely, if $\widehat{\btheta} = (\widehat{\bbeta}, \widehat{\bpi}, \widehat{\brho})$ solves \eqref{mixreg:commplog-varsel1}, then  $\widehat{\bvartheta} = (\widehat{\bphi}, \widehat{\bpi}, \widehat{\brho})$ solves \eqref{mixreg:commplog-varsel2} where $\widehat{\bphi}$ is such that $\widehat{\bphi}_{j} = \widehat{\bbeta}_0+\widehat{\bbeta}_j$, $j=1,\ldots,m$.
\end{itemize}
\end{proposition}

It turns out that \eqref{mixreg:commplog-varsel1} is more convenient to use in computation, while \eqref{mixreg:commplog-varsel2} is more useful in the theoretical investigation. {\revcolor{red}We also show that these penalized estimation criteria avoids the unbounded likelihood problem \citep{mclachlan2004finite} in Section \ref{app:boundedness} of Supplementary Materials.}


\section{Asymptotic Properties}\label{sec:th}

The generalized lasso formulation allows us to perform the asymptotic analysis under the unconstrained mixture regression model setup given in \eqref{eq:mixreg2}. The main issue is then in dealing with the special form of the generalized lasso penalty in \eqref{mixreg:commplog-varsel2}. To make things clear, we use $\btheta^*$ or $\bvartheta^*$ to denote the true parameters. We have defined $\mathcal{S}_R$ and $\mathcal{S}_H$ as the sets of relevant predictors and the predictors of sources of heterogeneity, respectively. Correspondingly, define 
\[
\mathcal{S} =  \{i; ((\I_p\otimes \A)\bphi^*)_i \neq 0\}.
\]
Recall that $\bbeta^* = \mbox{vec}(\B{B}\strans) = (\I_p \otimes \A)\bphi^*$, which means that $\mathcal{S}$ encodes the sparsity pattern of all the regression coefficients $\bbeta^*$ in the effects models. Then the recovery of $\mathcal{S}_R$ and $\mathcal{S}_H$ is immediate if $\mathcal{S}$ can be recovered.

We consider the setup that the design is random, and the number of predictors $p$ and the number of components $m$ are considered as fixed as the sample size $n$ grows. Building upon the works by \citet{fan2001}, \citet{stadler2010} and \citet{she2010gl}, our main results are presented in the following two theorems.

\begin{theorem}[Non-adaptive Estimator]\label{thm:const_estim}
Consider model \eqref{eq:mixreg2} with random design, fixed $p$ and $m$. 
Choose $\lambda = O(n^{-1/2})$. Assume the regularity conditions (A)-(C) from Section \ref{app:proofs} of Supplementary Materials on the joint density of $(y,\x)$ hold. Then for $\gamma= 0$, there exists a local maximizer $\widehat{\bvartheta}_{\lambda}^{\gamma}$ of \eqref{mixreg:commplog-varsel2} such that 
$
\sqrt{n}(\widehat{\bvartheta}_\lambda^{\gamma} - \bvartheta^*) = O_p(1) 
$.
\end{theorem}

\begin{theorem}[Adaptive Estimator]\label{thm:const_select}
Consider model \eqref{eq:mixreg2} with random design, fixed $p$ and $m$. 
Choose $\sqrt{n}\lambda \rightarrow 0$, $n^{(\gamma+1)/2}\lambda \rightarrow \infty$ as $n\rightarrow \infty$, and suppose the initial estimator in constructing the weights is $\sqrt{n}$-consistent, i.e., $\sqrt{n}(\widehat{\bvartheta}_\lambda^{ini} - \bvartheta^*) = O_p(1)$. Assume the regularity conditions (A)-(C) from Section \ref{app:proofs} of Supplementary Materials on the joint density of $(y,\x)$ hold. Then for any $\gamma > 0$, there exists a local maximizer $\widehat{\bvartheta}_{\lambda}^{\gamma}$ of \eqref{mixreg:commplog-varsel2} such that it is $\sqrt{n}$-consistent and
$P(\widehat{\mathcal{S}}_\lambda^{\gamma} = \mathcal{S}) \rightarrow 1$ as $n \rightarrow \infty$.
\end{theorem}
 
Theorem \ref{thm:const_estim} shows that the non-adaptive estimator can achieve $\sqrt{n}$-consistency in model estimation, under typical regularity conditions on the joint density of $(y,\x)$. Theorem \ref{thm:const_select} shows that the adaptive estimator, under the same conditions and with weights constructed from a consistent estimator such as the non-adaptive one in Theorem \ref{thm:const_estim}, can further achieve consistency in feature selection and heterogeneity pursuit.


\section{Computation}\label{sec:comp}

We propose a generalized EM algorithm for optimizing the criterion in \eqref{mixreg:commplog-varsel1}, which enjoys desirable convergence guarantee that the object function is monotone ascending along the iterations. The algorithmic structure is mostly straightforward based on the work of \citet{stadler2010}, except that in the M-step we need to efficiently solve an $\ell_1$ regularized weighted least squares problem with equality constraints. A Bregman coordinate descent algorithm \citep{Goldstein2009} is proposed to solve it. For tuning the number of component $m$ and the penalty parameter $\lambda$, we propose to minimize a Bayesian information criterion (BIC). To save space, the derivations of the algorithm and the details on tuning are provided in Section \ref{app:comp} of Supplementary Materials.


\section{Simulation}\label{sec:sim}

We compare the following methods via simulation,
\begin{itemize}
\item Normal mixture regression with variable selection via lasso (Mix-L, or M1) and via adaptive lasso (Mix-AL, or M2), proposed by \citet{stadler2010}.
\item The proposed normal mixture effects regression with variable selection and heterogeneity pursuit via lasso (Mix-HP-L, or M3) and via adaptive lasso (Mix-HP-AL, or M4).
\end{itemize}

The sample size is set to $n=200$ and the number of components is set to $m=3$. The data on the predictors, $\x_i \in \mathbb{R}^p$ for $ i = 1, \ldots, n$, are generated independently from multivariate normal distribution with mean $\mathbf{0}$ and covariance matrix $\bSigma$. We consider two correlation structures, i.e., the uncorrelated case with $\bSigma=\I_p$, and the correlated case with $\sigma_{ij} = 0.5^{|i-j|}$ where $\sigma_{ij}$ denotes the $(i,j)$'s entry of $\bSigma$. We consider three predictor dimensions: $p \in \{30, 60, 120\}$. As such, the number of free model parameters is 95, 185, and 365, respectively, being either comparable or much larger than the sample size.

In each setting, the first $p_0= 10$ predictors are relevant, and among them only $p_{00} = 3$ predictors have scaled heterogeneous effects over different components according to Definition \ref{def:2}. Specifically, under the mixture effects model~\eqref{eq:mixreg3} with $m=3$, the sub-vectors of the first 10 entries of the scaled coefficient vectors $\bbeta_j$, denoted as $\bbeta_{j0}$, $j = 0, 1, 2, 3$, are set as
\begin{align*}
  &\bbeta_{00} = (1, 1, 1, 1, 1, 1, 0, 0, 0)\trans / \sqrt{\delta}, \quad
    \bbeta_{10} = (0, 0, 0, 0, 0, 0, 0, 0, -3, 3)\trans / \sqrt{\delta},\\
  &\bbeta_{20} = (0, 0, 0, 0, 0, 0, 0, -3, 3, 0)\trans / \sqrt{\delta}, \quad
    \bbeta_{30} = (0, 0, 0, 0, 0, 0, 0, 3, 0, -3)\trans / \sqrt{\delta},
\end{align*}
and the variance components are set as $(\sigma_1^2, \sigma_2^2, \sigma_3^2)\trans = \delta \times (0.1,
    0.1, 0.4)\trans$, where $\delta$ controls the signal to noise ratio (SNR) defined as
$\mbox{SNR} = \sum_{j=1}^m \pi_j \bb_j\trans \mbox{cov}(\bX)\bb_j / \sum_{j=1}^m \pi_j\sigma_j^2$,
with $\bb_j = (\bbeta_0 + \bbeta_{j}) \times \sigma_j$, $j = 1, \ldots, m$ being the corresponding unscaled coefficient vectors as in the mean model \eqref{eq:mixreg1}. We remark that $\bb_j$s remain the same for different $\delta$ values, for facilitating the comparison among different SNRs. We choose $\delta = 1/8, 1/4, 1/2, 1, 2$, corresponding to $\mbox{SNR} = 200, 100, 50, 25, 12.5$, respectively. We set $\pi_1 = \pi_2 = \pi_3 = 1/m$ and generate the response values by \eqref{eq:mixreg3}. We choose the tuning parameter $\lambda$ and the number of components $m \in \{2, 3, 4\}$ by minimizing BIC. The experiment is repeated 500 times under each setting.

The following performance measures are computed. The estimation performance for the unscaled regression coefficients $(\bb_1, \ldots, \bb_m)$, the mixing probability $(\pi_1, \ldots, \pi_m)$ and the variances $(\sigma_1^2, \ldots, \sigma_m^2)$ is measured by their corresponding mean squared errors (MSE). The variable selection performance is measured by the false positive rate (FPR) and the true positive rate (TPR) for identifying relevant predictors, and the false heterogeneity rate (FHR) for identifying predictors with heterogeneous effects. Specifically, they are defined as below:
\begin{itemize}
\item FPR = \#falsely selected variables with no effects / \#variables with no effects;
\item TPR = \#correctly selected variables with effects / \#variables with effects;
\item FHR = \#falsely selected variables with {\revcolor{red}heterogeneous} effects / \#variables with common effects.
\end{itemize}

Figure~\ref{fig:simulation} displays the boxplots of mean squared errors in various simulation settings, and Table~\ref{tab:n=60} shows the detailed results for $p=60$ with $\bSigma = \bI_p$. {\revcolor{red}The results for $p \in \{30, 120\}$ and for the cases of correlated predictors convey similar messages, which are provided in Section \ref{app:sim} of Supplementary Materials.} The findings are summarized as follows.
\begin{itemize}
\item As expected, in general the larger the signal to noise ratio and the smaller the model dimensions, the better the performance of each method. 
\item Adaptive method in general leads to more accurate results in both model estimation and variable selection than its non-adaptive counterpart. The improvement can be substantial. {\revcolor{red}Specifically, both Mix-HP-L and Mix-HP-AL rarely miss important variables, but the former tends to select a larger model with more irrelevant variables. Indeed, the over-selection property of $\ell_1$ penalization is well known.}

\item The proposed methods Mix-HP-L and Mix-HP-AL outperform their counterparts without heterogeneity pursuit, Mix-L and Mix-AL, respectively, in most simulation setups, except that when $\mbox{SNR} = 12.5$ and $p=120$, all methods suffer from very low signal to noise ratio and very high dimensionality. The Mix-HP-AL has the best performance among all the competing methods; its improvement over others can be substantial especially when the signal is weak or moderate and the model dimension is high; moreover, in those relatively difficult scenarios, even Mix-HP-L can outperform Mix-AL. 
\item We have examined settings where all relevant predictors have heterogeneous effects, for which the methods with or without heterogeneity pursuit perform similarly. {\revcolor{red}We have also considered settings with unequal mixing probabilities, where the implications are similar; see Section \ref{app:sim} of Supplementary Materials. These results clearly demonstrate the benefit of heterogeneity pursuit, as it enables the potential of identifying the most parsimonious model.}
\end{itemize}

We conclude that overall the proposed {\red heterogeneity pursuit approach with adaptive lasso} (Mix-HP-AL) is preferable to both the non-adaptive {\red counterpart} Mix-HP-L and the conventional methods like Mix-L and Mix-AL. The proposed method is particularly beneficial when it is believed that only very few predictors contribute to the regression heterogeneity.

\begin{figure}
  \includegraphics[width=\textwidth]{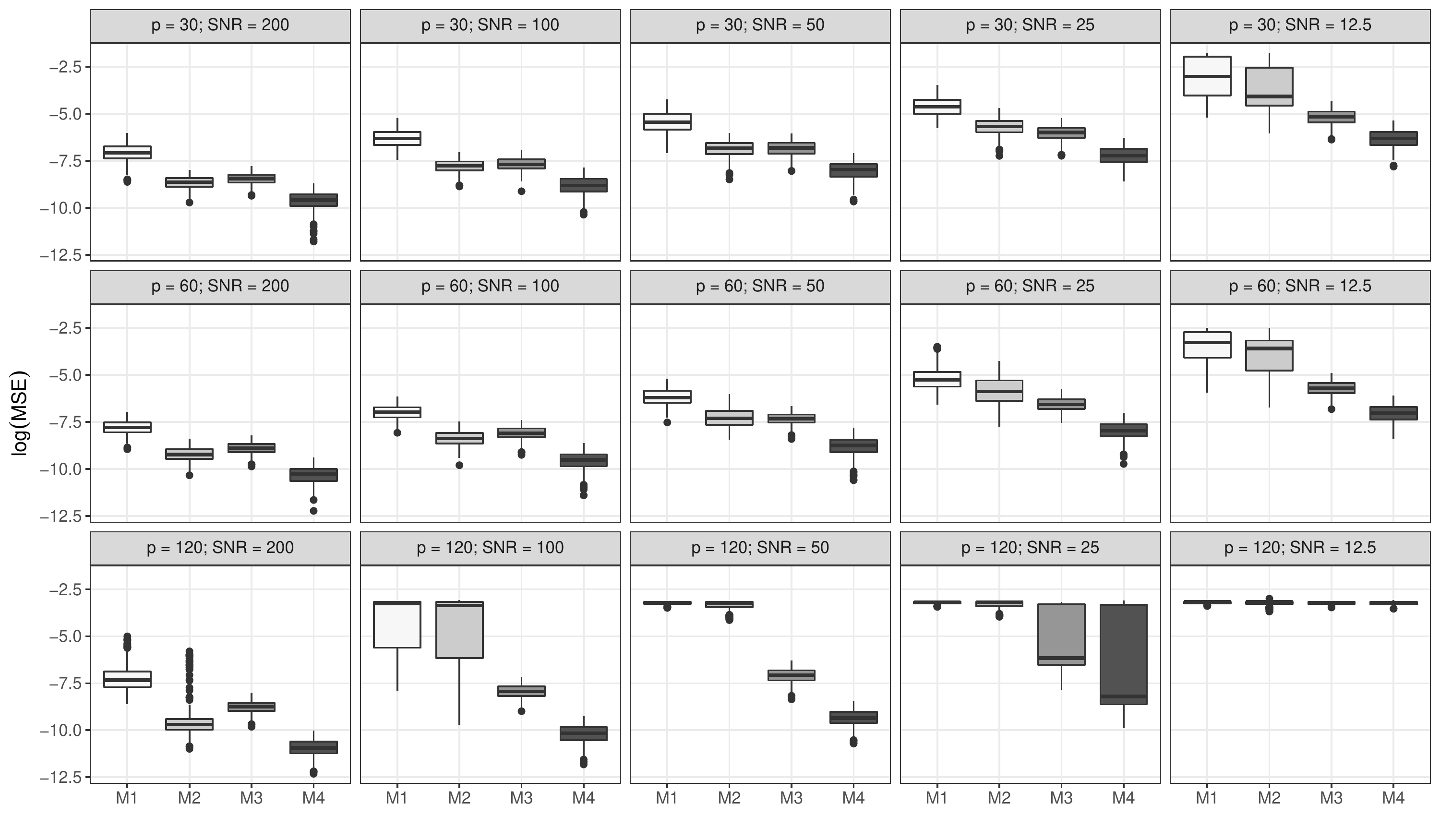}
  \caption{Boxplots of mean squared errors (in log scale) for
    estimating the unscaled coefficient vectors, for simulation
    settings with {\revcolor{red}$n=200$, $p \in \{30, 60, 120\}$, and $\bSigma = \bI_p$}, and
    $\mbox{SNR} \in \{200, 100, 50, 25, 12.5\}$. Four methods are
    compared: Mix-L (M1), Mix-AL (M2), Mix-HP-L (M3) and Mix-HP-AL
    (M4). The $\log (\mbox{MSE})$: logarithm of mean squared
    errors.}\label{fig:simulation}
\end{figure}

\begin{table}[tbp]
  \centering
  \caption{Comparison of mean squared error of estimation, variable
    selection and heterogeneity pursuit performance of four methods,
    Mix-L, Mix-AL, Mix-HP-L and Mix-HP-L, under settings with {\revcolor{red}$n = 200$, $p =
    60$, and $\bSigma = \bI_p$}. The mean squared errors (MSE) are reported along with their standard errors in the parenthesis. The simulation is based on $500$ replications. The MSE values are scaled by multiplying 100, and the FPR, FHR, TPR values are reported in percentage. } \label{tab:n=60}
\begin{tabular}{llrrrrrr}
  \toprule
  & & \multicolumn{3}{c}{MSE} & \multicolumn{3}{c}{RATE}\\
  \cmidrule(lr){3-5} \cmidrule(lr){6-8}
  $\mbox{SNR}$ & Method & $\bb$ & $\bsigma^2$ & $\bpi$ & FPR & FHR & TPR \\
  \midrule
  200 & Mix-L & 0.04 (0.02) & 7.32 (4.04) & 0.18 (0.15) & 53.0 & 100.0 & 100.0 \\
  & Mix-AL & 0.01 (0.00) & 0.09 (0.08) & 0.12 (0.10) & 10.2 & 100.0 & 100.0 \\
  & Mix-HP-L & 0.01 (0.00) & 4.95 (1.88) & 0.14 (0.11) & 39.1 & 4.4 & 100.0 \\
  & Mix-HP-AL & 0.00 (0.00) & 0.07 (0.05) & 0.11 (0.10) & 4.0 & 0.3 & 100.0 \\
  \midrule
  100 & Mix-L & 0.10 (0.04) & 10.58 (5.57) & 0.22 (0.18) & 48.3 & 100.0 & 100.0 \\
  & Mix-AL & 0.03 (0.01) & 0.10 (0.10) & 0.12 (0.11) & 13.3 & 100.0 & 100.0 \\
  & Mix-HP-L & 0.03 (0.01) & 7.04 (2.83) & 0.15 (0.12) & 36.8 & 3.8 & 100.0 \\
  & Mix-HP-AL & 0.01 (0.00) & 0.07 (0.06) & 0.12 (0.10) & 3.8 & 0.1 & 100.0 \\
  \midrule
  50 & Mix-L & 0.23 (0.11) & 14.66 (8.24) & 0.28 (0.23) & 43.9 & 100.0 & 100.0 \\
  & Mix-AL & 0.08 (0.05) & 0.22 (0.22) & 0.15 (0.13) & 16.6 & 100.0 & 100.0 \\
  & Mix-HP-L & 0.07 (0.02) & 9.91 (3.85) & 0.17 (0.14) & 33.4 & 2.8 & 100.0 \\
  & Mix-HP-AL & 0.02 (0.01) & 0.09 (0.10) & 0.13 (0.11) & 3.7 & 0.1 & 100.0 \\
  \midrule
  25 & Mix-L & 0.72 (0.60) & 29.07 (25.26) & 0.54 (0.61) & 33.9 & 100.0 & 100.0 \\
  & Mix-AL & 0.36 (0.25) & 1.34 (1.70) & 0.28 (0.32) & 19.0 & 100.0 & 100.0 \\
  & Mix-HP-L & 0.15 (0.06) & 12.57 (5.30) & 0.19 (0.15) & 33.2 & 1.9 & 100.0 \\
  & Mix-HP-AL & 0.04 (0.02) & 0.15 (0.16) & 0.13 (0.11) & 4.0 & 0.2 & 100.0 \\
  \midrule
  12.5 & Mix-L & 4.11 (2.58) & 101.20 (76.46) & 7.01 (5.90) & 16.0 & 100.0 & 90.0 \\
  & Mix-AL & 3.24 (2.56) & 33.17 (37.18) & 4.66 (4.71) & 12.3 & 100.0 & 90.0 \\
  & Mix-HP-L & 0.36 (0.13) & 16.50 (8.27) & 0.22 (0.18) & 35.0 & 1.9 & 100.0 \\
  & Mix-HP-AL & 0.10 (0.04) & 0.38 (0.40) & 0.15 (0.13) & 5.5 & 0.1 & 100.0 \\
  \bottomrule
\end{tabular}
\end{table}


\section{Applications}\label{sec:app}

\subsection{Alzheimer's Disease Neuroimaging Initiative (ADNI)}\label{sec:app-adni}
We performed imaging genetics analysis based on data from the ADNI database which is a public-private partnership to study the progression of mild cognitive impairment and early Alzheimer's disease based on different data sources including genetics, neuroimaging, clinical assessments, etc. (See \citeauthor{ADNI2020} for detailed study design and data collection information). Our goal here is to find out whether distinct clusters of disease-gene associations exist, possibly corresponding the disease stages, and to identify common genetic factors associated with overall disease risk, as well as cluster-specific ones. 

Briefly, to control data quality and reduce population stratification effect, we only include 760 Caucasian subjects in this analysis. For each subject, single nucleotide polymorphisms (SNPs)  genotyping data were acquired by Human 610-Quad BeadChip (Illumina, Inc., San Diego, CA) according to the manufacturer's protocols; and  raw MRI data were collected through 1.5 Tesla MRI scanners and then preprocessed by standard steps including anterior commissure and posterior commissure correction,  skull-stripping, cerebellum removing, intensity inhomogeneity correction, segmentation, and registration \citep{shen2004measuring}. The preprocessed brain images were further labelled regionally by existing template and then transferred following the deformable registration of subject images \citep{wang2011robust}, which eventually led to 93 regions of interest over whole brain.  After removing the ones with sex check failure, more than 10\% missing single nucleotide polymorphisms (SNPs), and outliers, 741 subjects including 174 Alzheimer's disease, 362 mild cognitive impairment and 205 healthy controls remain in the analysis.

We consider the following imaging phenotypes:  two global brain measurements, i.e. whole brain/white matter volumes, and two Alzheimer's disease related endophenotypes, i.e. left and right lateral ventricles volumes. For each imaging trait, we include both  SNPs belonging to the top 10 Alzheimer's disease candidate genes provided by the \citeauthor{AlzGene2020} database and those  identified from United Kingdom (UK) Biobank \citep{zhao2019gwas} ($\sim$20,000 subjects)  under the same imaging phenotype. The final lists of SNP names are provided in online supplementary materials. 
We fit our proposed Mix-AL-HP model for each imaging trait and its corresponding genetic predictors to examine the cluster patterns and select risk factors that impact the whole cohort with common effects and those impact the sub-groups/clusters heterogeneously. Age, gender and the top five genetic principal components are always included in the models as controls with common effects and no regularization. We fit models with different component numbers ($m \in \{1, 2, \ldots, 5\}$) and with/without the assumption of equal variance; the best model is selected based by BIC.

We first examine the identified clusters for each imaging trait to see whether the cluster pattern is associated with disease progression. 
The numbers of clusters for the four imaging traits, left/right ventricles and whole brain/white matter volumes, are 2, 3, 3, 1 with the smallest BIC values regarding $\lambda$ being 1873.12, 1871.58, 1794.04 and 2144.72, respectively. And among the three imaging traits with more than one identified cluster, the average values of imaging phenotype are shown to be clearly different over different clusters.
See Figure \ref{fig:adni_boxplot} in Section \ref{app:adni} of Supplementary Materials, which shows the cluster-specific boxplots of each imaging trait.
Intriguingly, given the fact that the size of brain increases along Alzheimer's disease progression, we are able to clearly align the identified clusters to different disease stages in light of the average volume of imaging traits. 
Note that for the white matter volume, which is a global brain phenotype, no cluster pattern is detected, which is biologically reasonable due to its weaker pathological bounding to disease etiology.

All the selected SNPs and their types of effect (common or cluster-specific) are summarized in Figure \ref{fig:adni_heatmap}. Most of the identified genetic risk variants (e.g. SNPs within genes CD2AP, MRVI1, GNA12) associated with two Alzheimer's disease imaging biomarkers are consistent and subtype-related, indicating the existence of varying genetic effects on brain structure over diesease progression. Meanwhile, the selected SNPs related to global brain phenotypes are generally with common effect; again, this is due to their weaker pathological bounding to disease etiology compared with the Alzheimer's disease related endophenotypes. 

Figure \ref{fig:adni_coefmap} provides a visualization of the estimated coefficients of each selected SNP under different clusters. Based on Figure \ref{fig:adni_coefmap}, we successfully detect a few SNPs showing a particular strong impact on early- to middle-stage Alzheimer's disease including rs2025935,  rs677909 and  rs798532 located in genes BIN1, MS4A4E and GNA12. Among them, BIN1 is the key molecular factor to modulate tau pathology and has recently been recognized as an important risk locus for late-onset Alzheimer's disease \citep{tan2013bridging}; MS4A4E has been detected by GWAS as a genetic risk factor for Alzheimer's disease based on Alzheimer's Disease Genetic Consortium \citep{hollingworth2011common}; and GNA12, though has not been extensively reported in existing experiments, is known to over-express in human brain.  Due to a typical small/moderate effect of single genetic signal, some of these variants are highly likely to be buried under existing methods without clustering overall heterogeneity. Moveover, our results provide valuable insights to prioritize future early therapeutic strategies even among all the Alzheimer's disease genotypes. In terms of other selected SNPs, most of them have been recognized as Alzheimer's disease risk factors in previous experiments or analyses, and they either show a common effect across all the clusters or a mixing one including both early and late stages in our results. Detailed estimation results are reported in Section \ref{app:adni} of Supplementary Materials.

\begin{figure}[tbp]
  \begin{center}\includegraphics[width=0.8\textwidth]{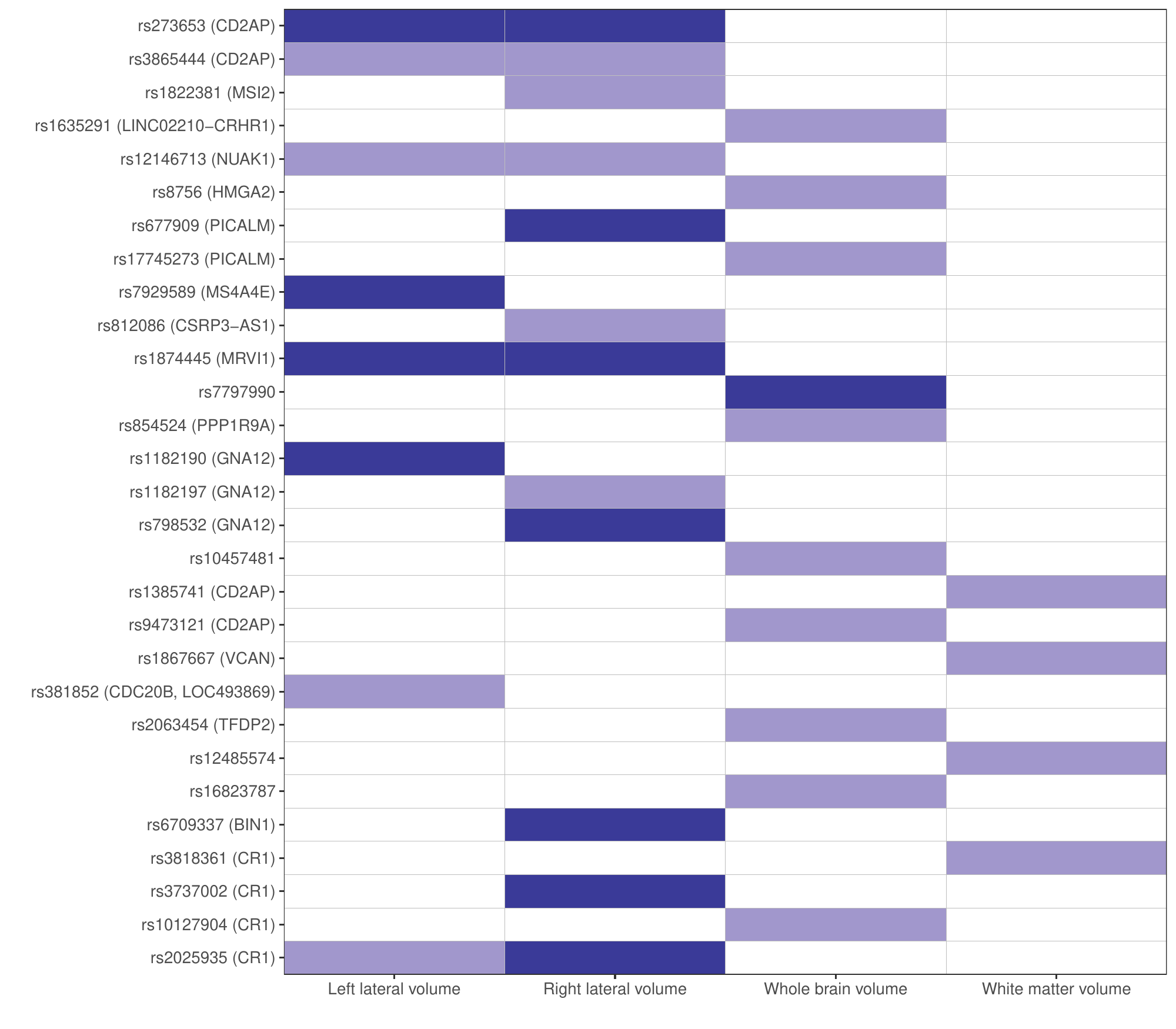}\end{center}
  \caption{ADNI study: effects of selected SNPs and their associated genes for the four imaging phenotypes. Light color means a SNP has only common effect across clusters; Dark color means a SNP has different effects across clusters and thus is considered as a source of heterogeneity. The SNPs are ordered based on the their positions on chromosomes. This figure appears in color in the electronic version of this article, and any mention of color refers to that version.}\label{fig:adni_heatmap}
\end{figure}

\begin{figure}[tbp]
  \begin{center}\includegraphics[width=\textwidth]{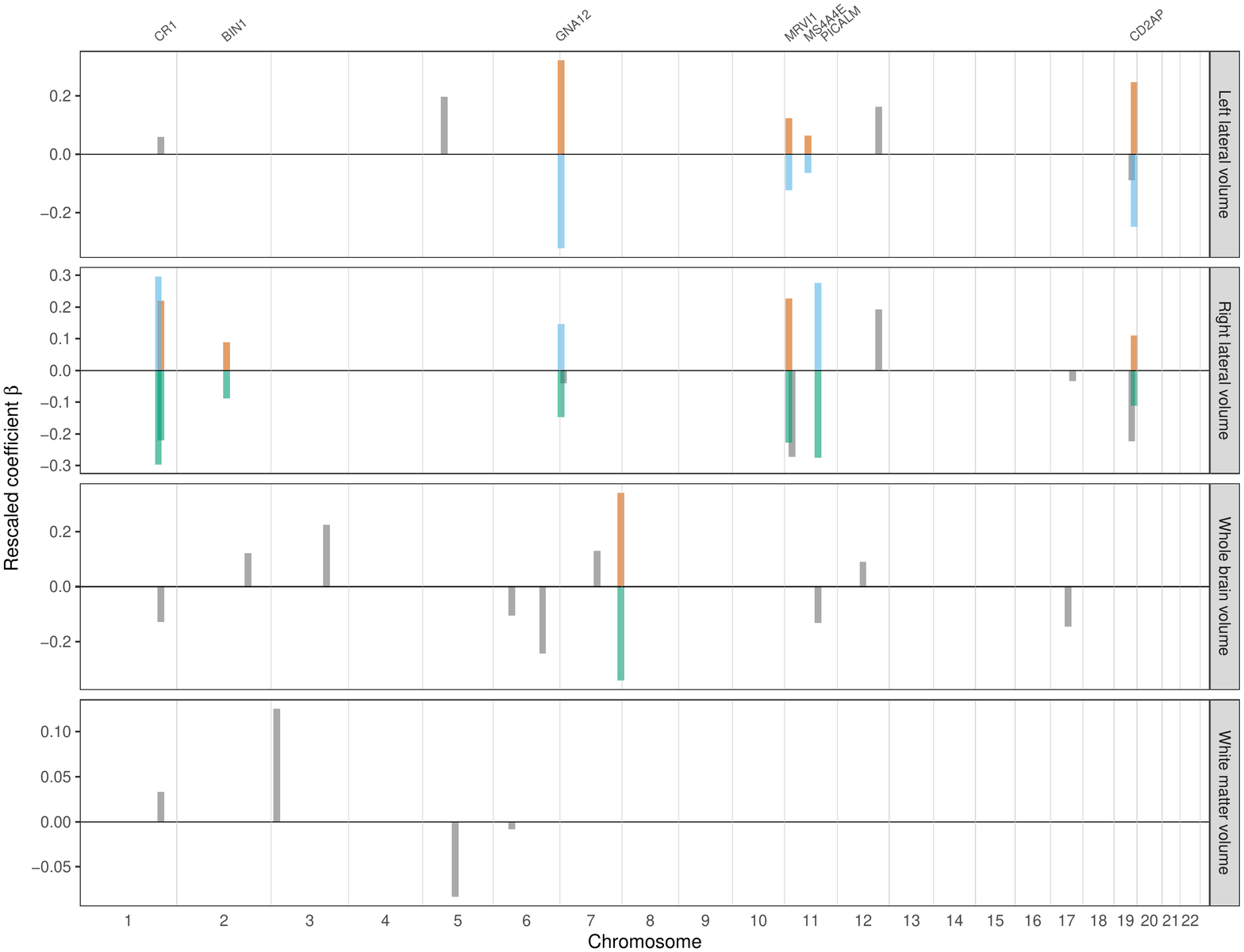}\end{center}
  \caption{ADNI study: estimated scaled coefficients ($\widehat{\phi}_{kj}$) of selected SNPs for the four imaging phenotypes, showing along their positions on chromosomes. The numbers of clusters are 2,3,3,1 for the four imaging phenotypes, showing from top to bottom. For each imaging phenotype, its cluster labels are aligned with decreasing average values of the phenotype (thus correspond to different disease stages).
  Grey color means a SNP has only common effect across clusters; red color indicates cluster 1; blue color indicates cluster 2; and green color indicates cluster 3. This figure appears in color in the electronic version of this article, and any mention of color refers to that version.}\label{fig:adni_coefmap}
\end{figure}

\subsection{Connecticut Adolescent Suicide Risk Study}\label{sec:app-suicide}

Suicide among youth is a serious public health problem in the United States. The Centers for Disease Control (CDC) reported that suicide is the third leading cause of death of youth aged 15–24 based on 2013 data, and more alarmingly there has been an increasing trend over time. Suicide prevention among youth is a very challenging task, which requires a systematic approach through developing reliable metrics for assessing suicide risk, locating areas of greater risk for effective resource allocation, identifying important risk factors, among others.

We use data from the State of Connecticut at the school district level to explore the association between suicide risk among 15-19 year olds and the characteristics of their school district. Specifically, the overall suicide risk of the 15-19 age group within each school district is proxied by its log-transformed 5-year average rate of inpatient hospitalizations due to suicide attempts from 2010 to 2014 (per year per 10,000 population). Several characteristics of the $n=119$ school district characteristics were collected in the same period: (1) demographic measures, including percent of households that included an adult male, average household size, percent of the population that are under 18 years of age, percent of population who are White; (2) academic measures, including average score on the Connecticut Academic Performance Test (CAPT), graduation rate, dropout rate, and attendance rate of high schools in the district; (3) behavioral measures, including incidence rate, defined as the ratio between the number of disciplinary incidences and the total enrollment, and serious incidence rate; (4) economic measures, including median income and free lunch rate. More details about the data can be found in \citet{Chen2017suicide}.

In the previous study, a generalized mixed-effects model was used to estimate the common effects of the school district characteristics on the suicide risk (through fixed-effects terms) and identify the ``overachievers'' and ``underachievers'' (through district-level random effects) among school districts. (It was also shown that there was no significant spatial effect.) Indeed, the existence of these anomalous school districts suggests that the regression association may not be homogeneous, and thus it is interesting to see whether additional insights can be gained by a mixture regression analysis, to reveal the potentially heterogeneous association structure, cluster the school districts, and identify the district characteristics that drive the heterogeneity. We thus apply our proposed Mix-HP-AL method to analyze the data. For dealing with the highly-correlated school district measurements, we perform group-wise principal component analysis and use each leading factor to summarize the information of each category, which results in $p=4$ district factors; the details of the principal component analysis results are provided in Section \ref{app:suicide} of Supplementary Materials.

Table~\ref{tab:res_suicide_0} reports the estimation results, and Figure~\ref{fig:cluster_suicide_hp} shows the corresponding cluster pattern of the school districts using the naive Bayes classification rule with the estimated component probabilities $\widehat{p}_{ij}$, i.e., $\widehat{z}_{ik} = 1$ if $k = \arg\max_{j} \widehat{p}_{ij}$. A three-component model is selected based on the BIC, in which the three factors differentiate school districts not in terms of their overall suicide risk as we did in our prior analysis, but in terms of the association between the risk factors and suicide risk. In Table~\ref{tab:res_suicide_0} one can see that only the demographic and academic factors are selected; when conditioning on the selected factors, the economic factor and the behavior factor are no longer related to the suicide risk, which may be partly due to the fact that the four factors are still moderately correlated. The major difference among the 3 clusters of communities involves the direction of effects of the demographic factor, which indicates a great deal of heterogeneity in how this factor impact suicide risk across communities. The majority of the school districts are in cluster 3, in which the suicide risk is negatively associated with the demographic factor; that is, in general, the larger the household size, the greater percentage of households with an adult male, the greater the proportion of population under age 18, and higher the proportion of White residents are associated with {\it lower} suicide risk, after adjusting for the effect of academic performance. In contrast, in cluster 1 the association between the suicide risk and the demographic factor is positive, such that higher rates of male householders, larger household size, higher proportions of children under 18, and a higher proportion of White residents is associated with {\it higher} suicide risk. Further analysis reveals that the 12 school districts in cluster 1 have significantly lower mean suicide risk than those in cluster 3; this suggests that the impact of the demographic factors on suicide risk may change and even flip sign with the mean suicide risk level itself. It is possible that this is caused by some ``unmeasured'' factors confounded with the demographic factor. Cluster 2 is the smallest in size among the three, consisting of ``Regional 19'' (near the University of Connecticut), ``New London'' and ``Monroe''; these are anomalous districts with very low suicide risk. The academic factor, in contrast, is identified to have only common effects after scaling by the variances according to Definition \ref{def:2}, which makes the estimated model even more parsimonious. That the effect of the academic factor is always positive indicates that suicide risk tends to be higher in those school districts with better academic performance; as discussed in \citet{Chen2017suicide}, students in school districts of better academic performance could be under higher pressure, which in turn may induce more psychological distress. In general, our results agree well with previous studies, and we gain additional insight on the changing impact of the school district characteristics on the suicide risk.

\begin{table}[tbph]
  \centering \caption{Suicide risk study: the coefficient estimates using Mix-HP-AL. The zero values are shown as blanks.}
  \label{tab:res_suicide_0}
  \begin{tabular}{lrrrrrrr}
    \toprule
    Factors & $\widehat{\bphi}_1$ & $\widehat{\bphi}_2$ & $\widehat{\bphi}_3$\\
    \midrule
    Intercept & 6.23 & 6.23 & 6.23 \\
    Demographic factor & 0.27 & & -0.27 \\
    Academic factor & 0.13 & 0.13 & 0.13 \\
    Behavioral factor  & & & \\
    Economical factor & & & \\
    $\widehat\sigma$ & 0.33 & 0.20 & 0.46 \\
    $\widehat\pi$ & 0.13 & 0.02 & 0.85 \\
    \bottomrule
  \end{tabular}
\end{table}

We have also compared Mix-HP-AL to Mix-AL by performing a random-splitting procedure to evaluate their out-of-sample predictive performance. Each time the data is split to 80\% training for model fitting and 20\% testing for out-of-sample evaluation, and the procedure is repeated 500 times. The average predictive log-likelihood (with standard error in the parenthesis) is $-24.6$ ($2.32$) and $-23.2$ ($2.16$) for Mix-AL and Mix-HP-AL, respectively, indicating that the proposed method indeed performs better for this dataset.

\begin{figure}[htp]
  \begin{center}\includegraphics[width=0.9\textwidth]{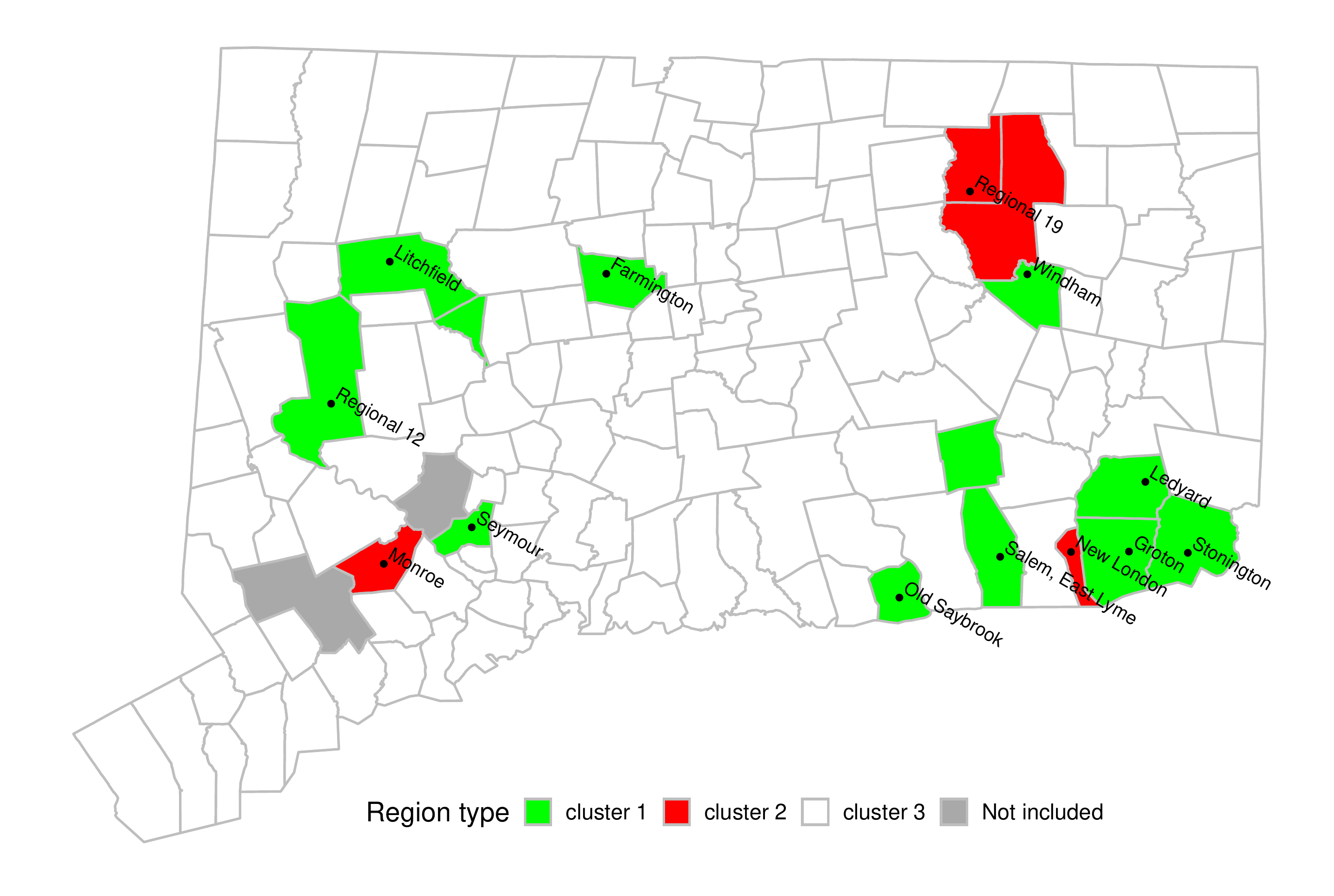}\end{center}
  \caption{Suicide risk study: district clustering using Mix-HP-AL. This figure appears in color in the electronic version of this article, and any mention of color refers to that version.}\label{fig:cluster_suicide_hp}
\end{figure}

The proposed method has also been applied to another application in sports analytics for understanding how the salaries of baseball players are associated with their performance and contractual status. Due to space limit, the application is detailed in Section~\ref{app:baseball} of Supplementary Materials.


\section{Discussion}\label{sec:dis}

In this paper, we propose a mixture regression method to thoroughly explore the heterogeneity in a population of interest, which is increasingly encountered in the era of big data. Our approach goes beyond the conventional variable selection methods, by not only identifying the relevant predictors, but also distinguishing from them the true sources of heterogeneity. As such, the proposed approach can potentially lead to a much more parsimonious and interpretable model to facilitate scientific discovery.

There are a number of future research directions. It is pressing to extend the proposed method to handle non-Gaussian outcomes, such as binomial mixture and Poisson mixture. This extension can help us to improve the analysis for the suicide risk study, as the raw counts of the suicide-related hospital admissions may be better modeled by Poisson distribution. {\revcolor{red}Another possible direction is to consider other forms of penalty functions. For example, when the predictors are highly correlated, it could be beneficial to use the elastic-net penalty \citep{zou2005regularization} to ensure stable coefficient estimation. Non-convex penalties could also be considered to improve variable selection.} A related task is to extend the theoretical analysis to high-dimensional settings where the number of variables may grow with or exceed the sample size. A potential byproduct of the proposed approach is that it can lead to automatic reduction of the number of pre-specified clusters when the effects of some clusters are estimated to be exactly the same; it is hopeful that this interesting feature can allow us to build a more general mixture learning framework where relevant variables, sources of heterogeneity and the number of clusters are simultaneously learned. It would also be interesting to consider heterogeneity pursuit in multivariate mixture regression, but it is not straightforward. The mixture components may have different covariance matrices which complicate the definition of the sources of heterogeneity, and the set of predictors with heterogeneous effects may differ across different responses.

In this work, we mainly focus on the framework of mixture regression to pursue the sources of heterogeneity at the ``global" level. An interesting direction is to extend our work to utilize the frameworks of individualized modeling and sub-group analysis which mainly pursue the sources of heterogeneity at the ``individual" level \citep{tang2020individual}. To lessen the assumptions of mixture regression, several recent works formulate the problem as a penalized regression with a fusion-type penalty. \citet{Ma2017fusion} proposed a concave pairwise fusion approach to identify sub-groups with pairwise penalization on subject-specific intercepts. \citet{Austin2017fusion} proposed a grouping fusion approach to identify unknown sub-groups and their corresponding regression models. \citet{tang2020individual} proposed a method to simultaneously achieve individualized variable selection and sub-grouping. Comparing to the mixture model framework, an individualized penalized regression approach may not fully utilize the potential global mixture structure and fails to consider the potential heterogeneity in variances. Therefore, we will explore the idea of combining mixture model and individualized fusion, to simultaneously perform global and individualized heterogeneity pursuits.


\if1\blind {
\bigskip
\section*{Acknowledgments}
Yu is supported by National Natural Science Foundation of China grant 11661038 and Jiangxi Provincial Natural Science Foundation grant 20202BABL201013. Yao is supported by U.S. National Science Foundation grant DMS-1461677 and the Department of Energy award No. 10006272. Aseltine is supported by the U.S. National Institutes of Health grant R01-MH112148, R01-MH112148-03S1, and R01-MH124740. Chen is supported by the U.S. National Science Foundation grants DMS-1613295 and IIS-1718798, and the U.S. National Institutes of Health grants R01-MH112148, R01-MH112148-03S1, and R01-MH124740.

ADNI data used in preparation of this article were obtained from the Alzheimers Disease Neuroimaging Initiative (ADNI) database (http://adni.loni.usc.edu). As such, the investigators within the ADNI contributed to the design and implementation of ADNI and/or provided data but did not participate in analysis or writing of this paper.
} \fi

\section*{Supplementary Materials}

\setcounter{figure}{0}
\setcounter{table}{0}
\renewcommand{\thefigure}{\thesection.\arabic{figure}}
\renewcommand*{\theHfigure}{\thesection.\arabic{figure}}
\renewcommand{\thetable}{\thesection.\arabic{table}}
\renewcommand*{\theHtable}{\thesection.\arabic{table}}


\begin{appendices}

\section{Proofs of the Main Theorems}\label{app:proofs}

\subsection{Regularity Conditions}\label{app:cond}

We use the same regularity conditions (A)--(C) as in \citet{fan2001} and \citet{stadler2010} for analyzing the proposed mixture regression method. Let $\bV_i = (y_i,\bx_i)$ be the $i{\mbox{th}}$ observation for $i = 1, \ldots, n$. Recall that $\bvartheta = (\bphi_1, \ldots,\bphi_m;\pi_1,\ldots,\pi_m;\rho_1,\ldots,\rho_m)$ collects all the unknown parameters from the model presented in \eqref{eq:mixreg2} of the main paper. Let
\begin{align*}
  \Omega = \left\{ \right.& (\bphi_1, \ldots,\bphi_m;\pi_1,\ldots,\pi_m;\rho_1,\ldots,\rho_m) \mid\\
                   & \bphi_j \in \mathbb{R}^p\mbox{ for }j = 1,\ldots,m;\\
                   & \pi_j > 0 \mbox{ for }j = 1,\ldots, m \mbox{ with } \sum_{j=1}^m\pi_j = 1;\\
&  \left.\rho_j >0\mbox{ for } j = 1,\ldots,m.\right\}
\end{align*}

\begin{condition}\label{app:cond:a}
  The observations $\bV_i$, $i = 1, \ldots, n$ are independent and identically distributed with joint probability density $f(\bV, \bvartheta)$ with respect to some measure $\mu$. $f(\bV, \bvartheta)$ has a common support and the model is identifiable. Furthermore, the first and second logarithmic derivatives of $f$ satisfying the equations
  \begin{align*}
    E_{\bvartheta}\left[\frac{\partial \log f(\bV, \bvartheta)}{\partial \bvartheta}\right] = \0
  \end{align*}
  and
  \begin{align*}
    \bI_{jk}(\bvartheta) =
    E_{\bvartheta} \left[\frac{\partial}{\partial \vartheta_j} \log f(\bV, \bvartheta) \frac{\partial}{\partial \vartheta_k} \log f(\bV, \bvartheta) \right] =
    E_{\bvartheta} \left[- \frac{\partial^2}{\partial \vartheta_j \vartheta_k} \log f(\bV, \bvartheta) \right].
  \end{align*}
\end{condition}

\begin{condition}\label{app:cond:b}
  The fisher information matrix
  \[
    \bI(\bvartheta) = E_{\bvartheta} \left\{ \left[\frac{\partial}{\partial \bvartheta} \log f(\bV, \bvartheta)\right] \left[\frac{\partial}{\partial \bvartheta} \log f(\bV, \bvartheta) \right]\trans \right\}
  \]
  is finite and positive definite at the true parameter vector $\bvartheta = \bvartheta^*$.
\end{condition}

\begin{condition}\label{app:cond:c}
  There exists an open subset $\omega$ of $\Omega$ that contains the true parameter vector $\bvartheta^*$ such that for almost all $\bV$ the density function $f(\bV, \bvartheta)$ admits all third derivatives for all $\bvartheta \in \Omega$. Furthermore, there exist functions $M_{jkl}(\cdot)$ such that
  \[
    \bigl\vert \frac{\partial^3}{\partial \vartheta_j \vartheta_k \vartheta_l} \log f(\bV, \bvartheta) \bigr\vert \le M_{jkl}(\bV) \mbox{ for all } \bvartheta \in \omega,
  \]
  where $E_{\bvartheta^*} [M_{jkl}(\bV)] < \infty$ for $i, j, l$.
\end{condition}

\subsection{Proof of Theorem \ref{thm:const_estim} of the Main Paper}

\begin{proof}
  We need to show that for any given $\epsilon >0$, there exists a large constant $c$, such that
\begin{align*}
\mathbb{P}\left\{\sup_{\bu:\|\bu\| = c} \l_\lambda(\bvartheta^* + n^{-1/2}\bu) < \l_\lambda(\bvartheta^*) \right\} \geq 1-\epsilon.
\end{align*}
{\revcolor{red}Here $\l_\lambda$ denotes the generalized lasso criterion defined in \eqref{mixreg:commplog-varsel2} of the main paper when $\gamma = 0$.} 

Define
\begin{align*}
D_n(\bu) & = \l_\lambda(\bvartheta^* + n^{-1/2}\bu) - \l_\lambda(\bvartheta^*)\notag\\
 & =  \l_0(\bvartheta^* + n^{-1/2}\bu) - \l_0(\bvartheta^*)
 - n\lambda \{\|(\I_p \otimes \A)(\bphi^*+n^{-1/2}\bu_{\phi}) \|_1 -  \|(\I_p \otimes \A)\bphi^* \|_1\}.
\end{align*}
Here $\bu_{\phi}$ is a subvector of $\bu$ corresponding to the parameter vector of the regression coefficients $\bphi$.

Denote $(\I_p \otimes \A)\bphi^* = \widetilde{\bphi}^*$ and $(\I_p \otimes \A)\bu_{\phi} = \widetilde{\bu}_{\phi}$. Recall that $
\mathcal{S} =  \{i; ((\I_p\otimes \A)\bphi^*)_i = \widetilde{\phi}^*_{i} \neq 0\}$; denote $\widetilde{\bphi}^*_{\mathcal{S}}$ as the subvector of $\widetilde{\bphi}^*$ and $\widetilde\bu_{\phi,\mathcal{S}}$ the subvector of $\widetilde\bu_{\phi}$ corresponding to the indices in $\mathcal{S}$. Then
\begin{align}
D_n(\bu)
 \leq & \l_0(\bvartheta^* + n^{-1/2}\bu) - \l_0(\bvartheta^*)
  - n\lambda \{\|\widetilde{\bphi}^*_{\mathcal{S}}+n^{-1/2}\widetilde{\bu}_{\phi,\mathcal{S}} \|_1 -  \|\widetilde{\bphi}^*_{\mathcal{S}}\|_1\}.\label{eq:A1}
\end{align}

By the regularity conditions in Section~\ref{app:cond},
\begin{align}
\l_0(\bvartheta^* + n^{-1/2}\bu) - \l_0(\bvartheta^*) = n^{-1/2}\l_0'(\bvartheta^*)\trans \bu - \frac{1}{2}\bu\trans\I(\bvartheta^*)\bu(1+o_p(1)),\label{eq:A2}
\end{align}
where $\l_0'(\cdot)$ denotes the score function so that $n^{-1/2}\l_0'(\bvartheta^*)=O_p(1)$, and $\I(\bvartheta^*)$ is the fisher information matrix at $\bvartheta = \bvartheta^*$ which is positive definite. As such, the second term, which is of the order $\|\bu\|^2$, dominates the first term uniformly in $\|\bu\| = c$ for $c$ sufficiently large.

It remains to show that the third term on the right hand side of \eqref{eq:A1} is also dominated by the second term in \eqref{eq:A2} for sufficiently large $\|\bu\|=c$.
We have that $\|\widetilde{\bu}_{\phi}\| \leq \|\I_p\otimes\A\|\|\bu_{\phi}\|$ and $\|\I_p\otimes\A\| = 1$ for any $p$ and $m$, where $\|\cdot\|$ of a matrix denotes its spectral norm, e.g., the largest singular value of the matrix \citep[See][pp. 71]{GoluVanl96}. To see this, observe that $\|\I_p\otimes\A\| = \|\A\|$, i.e., the largest singular value of $\A$ or the square root of the largest eigenvalue of $\A\trans\A$. Then by definition, we have that $\A\trans\A = \I_m - (1/m - 1/m^2)\J_m$ and $\|\A\| = 1$. It follows that for large enough $n$,
\begin{align*}
|n\lambda \{\|\widetilde{\bphi}^*_{\mathcal{S}} +n^{-1/2}\widetilde{\bu}_{\phi,\mathcal{S}} \|_1 -  \|\widetilde{\bphi}^*_{\mathcal{S}}\|_1\}|
  & = |n^{1/2}\lambda \widetilde{\bu}_{\phi,\mathcal{S}}\trans\mbox{sgn}(\widetilde{\bphi}^*_{\mathcal{S}})|\\
  & \leq n^{1/2}\lambda\sqrt{|\mathcal{S}|}\|\widetilde{\bu}_{\phi,\mathcal{S}}\|\\
  & \leq n^{1/2}\lambda\sqrt{|\mathcal{S}|} \|\bu_{\phi}\|.
\end{align*}
This completes the proof.
\end{proof}

\subsection{Proof of Theorem \ref{thm:const_select} of the Main Paper}

\begin{proof}
  The results are extended from those in \citet{she2010gl}, where the estimation consistency and sign consistency of clustered lasso were studied under linear regression setup. Therefore we briefly sketch the proof.

Consider the following general form of the proposed generalized lasso criterion
  \begin{align}
\label{mixreg:glasso}
\max_{\bvartheta} & \left\{\l_{0}(\bvartheta) -n\lambda \|\T\bphi\|_1
\right\}.
\end{align}
Here, in our problem, the matrix $\T$ takes the form $\T = (\I_p \otimes \A)$ for the non-adaptive estimation ($\gamma = 0$) and $\T = \W(\I_p \otimes \A)$ for the adaptive estimation ($\gamma>0$). Let $\T_{\mathcal{S}}$ be a submatrix of $\T$ by taking its rows corresponding to the index set
\[
\mathcal{S} = \{i; (\T\bphi^*)_i \neq 0\}.
\]
Similarly, define $\T_{\mathcal{S}^c}$ where $\mathcal{S}^c$ denotes the complement of $\mathcal{S}$. By the regularity conditions in Section~\ref{app:cond} and following the proof of Theorem 3.1 of \citet{she2010gl}, we can obtain a sufficient condition for the generalized lasso estimator from solving \eqref{mixreg:glasso} to achieve sign consistency,
\begin{align}
\|(\T_{\mathcal{S}^c}\I^{-1}(\bvartheta^*)\T_{\mathcal{S}^c}\trans)^{-}(\T_{\mathcal{S}^c}\I^{-1}(\bvartheta)\T_{\mathcal{S}}\trans)\|_{\infty} <1, \label{eq:irr}
\end{align}
where $\I(\bvartheta^*)$ is the Fisher information matrix, and $(\cdot)^-$ denote the Moore-Penrose inverse of the enclosed matrix.  This ``irrepresentable'' condition \citep{zhao2006} is rather strong and is hard to verity for the non-adaptive $\ell_1$ estimator. However, for an adaptive estimator with proper weights and tuning, the condition in \eqref{eq:irr} can be satisfied. Following the proof of Theorem 3.2 of \citet{she2010gl}, it boils down to verify that our choices of $\W$ and $\lambda$ satisfies that
\begin{align}
\frac{1}{\sqrt{n}\lambda}w_i^{-1} = o_p(1), \sqrt{n}\lambda w_{i'} = o_p(1),\label{eq:she}
\end{align}
for any $i \in\mathcal{S}^c$ and $i' \in \mathcal{S}$. Indeed, when
the construction of the weights in \eqref{eq:penalty} of the main paper is based on a $\sqrt{n}$-consistent initial estimator, it holds true that
\[
w_i^{-1} = O_p(n^{-\gamma/2}), w_{i'} = O_p(1).
\]
Given that $\sqrt{n}\lambda \rightarrow 0$, $n^{(\gamma+1)/2}\lambda \rightarrow \infty$ as $n\rightarrow \infty$, the results in \eqref{eq:she} follow immediately. This completes the proof.

\end{proof}

\clearpage
\section{Boundedness of the Penalized Log-likelihood Criteria of the Main Paper}\label{app:boundedness}

It is well known that maximum likelihood estimation in finite normal mixture model may suffer from the unbounded likelihood problem. However, in our proposed estimation criterion, the generalized lasso penalty term is a function of the re-scaled parameters $\bphi_j = \bb_j / \sigma_j = \bbeta_0 + \bbeta_j$ ($j=1,\ldots,m$), and hence small variances are penalized and discouraged. Following \citet{stadler2010}, we can readily show that our proposed estimation criteria avoids the unbounded likelihood problem. For the sake of completeness, we give the key results in the following proposition.

\begin{proposition}
  Assume that $y_i \ne 0$ for all $i = 1,\ldots,n$, and assume the weights $w_{jk} = |\widehat{\beta}^0_{j,k}|^{-\gamma}$ satisfies $0<w_{jk}<\infty$. Then the penalized log-likelihood criterion $\l^{\gamma}_\lambda(\bvartheta)$ defined in \eqref{mixreg:commplog-varsel2} of the main paper is bounded from above for all values $\bvartheta \in \Omega$ and $\lambda>0$. 
\end{proposition}

\begin{proof}
  \citet{stadler2010} studied the $\ell_1$ penalized mixture model using the re-scaled parameterization and showed the boundedness of the $\ell_1$ penalized log-likelihood criterion; see their Proposition 2 and Appendix C. To apply their results, we need to deal with the generalized lasso penalty form.

Consider the penalized log-likelihood criterion $\l^{\gamma}_\lambda(\bvartheta)$ defined in \eqref{mixreg:commplog-varsel2} of the main paper at $\gamma = 0$, denoted as $\l_\lambda(\bvartheta)$.
We can show that $\|(\I_p \otimes \A)\bphi\|_1 \ge 1/m\|\bphi\|_1$. To see this, {\revcolor{red}
observe that $\|(\I_p \otimes \A)\bphi\|_1 = \sum_{k=1}^p\|\A\widetilde \bphi_k\|_1$. Then based on the definition of $\A$, for $k=1,\ldots,p$,
\begin{align*}
\|\A\widetilde \bphi_k\|_1 
&= \frac{1}{m} \|\1_m\trans \widetilde\bphi_k\|_1 + \|\widetilde\bphi_k - \frac{1}{m}\J_m \widetilde\bphi_k\|_1\\
&\ge \frac{1}{m} \|\1_m\trans \widetilde\bphi_k\|_1 + \frac{1}{m}\|\widetilde\bphi_k - \frac{1}{m}\J_m \widetilde\bphi_k\|_1\\
&\ge \frac{1}{m} \|\1_m\trans \widetilde\bphi_k\|_1 + \frac{1}{m}\|\widetilde\bphi_k\|_1 - \frac{1}{m^2}\|\J_m \widetilde\bphi_k\|_1 \mbox{ (by triangle inequality)}\\
&=\frac{1}{m}\|\widetilde\bphi_k\|_1.
\end{align*}
Then $\|(\I_p \otimes \A)\bphi\|_1 \ge 1/m \sum_{k=1}^p\|\widetilde\bphi_k\|_1 = 1/m\|\bphi\|_1$.} It then follows that
\begin{align*}
l_\lambda(\bvartheta)
&\le \sum_{i=1}^n\log\left\{ f(y_i\mid \bx_i,\bvartheta) \right\}-
\frac{n\lambda}{m}\|\bphi\|_1.
\end{align*}
The right-hand side is exactly in the form of the $\ell_1$ penalized mixture log-likelihood considered in \citet{stadler2010}; its boundedness from above then directly follows. The extension to the adaptive version when $\gamma > 0$ directly follows since $\|\W(\I_p \otimes \A)\bphi\|_1 \ge \min\{w_{jk}\}\|(\I_p \otimes \A)\bphi\|_1$.

\end{proof}

Consequently, the constrained penalized log-likelihood form in  \eqref{mixreg:commplog-varsel1}, in terms of $\btheta=\A\bphi$, is also bounded from above, for any
\begin{align*}
  \btheta \in \left\{ \right.& (\bbeta_0,\bbeta_1,  \ldots,\bbeta_m;\pi_1,\ldots,\pi_m;\rho_1,\ldots,\rho_m) \mid\\
                   & \bbeta_j \in \mathbb{R}^p\mbox{ for }j = 0,\ldots,m; \sum_{j=1}^m\beta_{jk}=0\mbox{ for } k = 1, \ldots, p;\\
                   & \pi_j > 0 \mbox{ for }j = 1,\ldots, m \mbox{ with } \sum_{j=1}^m\pi_j = 1;\\
&  \left.\rho_j >0\mbox{ for } j = 1,\ldots,m.\right\}.
\end{align*}

\clearpage
\section{Details on Computation}\label{app:comp}
\subsection{Generalized Expectation-Maximization Algorithm}\label{app:gem}

We provide the detailed derivations of the proposed computational algorithm. Let $\bz_i = (z_{i1}, \ldots, z_{im})$, $i = 1, \ldots, n$ be the unobserved component labels, i.e., $z_{ij}=1$ if the $i{\mbox{th}}$
observation is from the $j{\mbox{th}}$ component, and $z_{ij}=0$ otherwise. Denote by $\by = (y_1, \ldots, y_n)\trans$ the vector of responses and by $\X = (\bx_{1},\ldots,\bx_{n})\trans$ the matrix of predictors. By \eqref{eq:mixreg3} of the main paper, the complete-data log-likelihood of $\{(y_{i},\bz_{i};\bx_{i}): i=1,\ldots,n\}$ is given by
\begin{align*}
  \ell_c(\btheta) \equiv \sum_{i=1}^n \sum_{j=1}^{m} z_{ij} \log \left\{ \pi_{j}\frac{\rho_j}{\sqrt{2\pi}}\exp\{-\frac{1}{2}(\rho_jy_i-\bx_i\trans\bbeta_0-\bx_i\trans\bbeta_{j})^2\} \right\}.\notag
\end{align*}

The Expectation-Maximization algorithm works by alternating between an E-step and a M-step until convergence is reached, e.g., the relative change in the parameters measured by $\ell_2$ norm is less than some tolerance level.

Denote the
parameter estimate at the $t{\mbox{th}}$ iteration as $\btheta^{(t)} = (\bbeta^{(t)}, \bpi^{(t)}, \brho^{(t)})$. In the E step, we compute the conditional expectation of the complete-data log-likelihood, 
\begin{align*}
  \mathbb{E}_{\btheta^{(t)}}(\ell_c(\btheta)\mid \by,\bX) = \sum_{i=1}^{n}\sum_{j=1}^{m}\widehat{p}_{ij} \log \left\{ \pi_{j} \frac{\rho_j}{\sqrt{2\pi}}\exp\{-\frac{1}{2}(\rho_jy_i-\bx_i\trans\bbeta_0-\bx_i\trans\bbeta_{j})^2\}\right\},
\end{align*}
where 
\begin{align}
\label{pij}
  \widehat{p}_{ij} = \mathbb{E}_{\btheta^{(t)}} (z_{ij} \mid \by, \bX) = \frac{\pi^{(t)}_{j} \rho^{(t)}_j \exp\{-\frac{1}{2}(\rho^{(t)}_jy_i-\bx_i\trans\bbeta^{(t)}_0-\bx_i\trans\bbeta^{(t)}_{j})^2\}}
  {\sum_{j=1}^{m}\pi^{(t)}_{j} \rho^{(t)}_j \exp\{-\frac{1}{2}(\rho^{(t)}_jy_i-\bx_i\trans\bbeta^{(t)}_0-\bx_i\trans\bbeta^{(t)}_{j})^2\}}.
\end{align}

In the M-Step, we deal with the constrained optimization of
\[
Q(\btheta \mid \btheta^{(t)}) \equiv -\mathbb{E}_{\btheta^{(t)}}(\ell_c(\btheta)\mid \by, \bX) + n\lambda\sum_{k=1}^{p}\calP(\widetilde{\bbeta}_{k}),
\]
with respect to $\btheta = \{\bbeta, \bpi, \brho\}$. The optimization in $\bpi$ is separable, i.e., 
  \begin{align*}
   & \bpi^{(t+1)} = \arg\min_{\bpi}\{- \sum_{i=1}^{n}\sum_{j=1}^{m}\widehat{p}_{ij} \log \pi_{j} \},\
   \mbox{subject to} \ \pi_j > 0 \ (j = 1, \ldots, m) \text{ and } \sum_{j=1}^m\pi_{j}=1, 
  \end{align*}
which yields 
  \[
    \pi^{(t+1)}_j = \frac{\sum_{i = 1}^n \widehat{p}_{ij}}{n} \quad (j  = 1, \ldots, m).
  \]
The problem then becomes 
  \begin{align}
    &\min_{\brho, \bbeta}\left\{
      - \sum_{i=1}^{n}\sum_{j=1}^{m}\widehat{p}_{ij} \log \rho_j +
      \frac{1}{2} \sum_{i=1}^{n}\sum_{j=1}^{m}\widehat{p}_{ij}(\rho_jy_i-\bx_i\trans\bbeta_0-\bx_i\trans\bbeta_{j})^2 +
      n\lambda\sum_{k=1}^{p}\calP(\widetilde{\bbeta}_{k})\right\}, \notag\\
    & \mbox{ subject to } \sum_{j=1}^m\beta_{jk}=0 \ (k = 1, \ldots, p) \text{ and } \rho_j > 0 \ (j = 1,  \ldots, m). \label{eq:obj_mstep}
  \end{align}
  We proceed by performing coordinate descent updates with two blocks of parameters $\brho$ and $\B{B} = (\bbeta_0,\bbeta_1, \ldots,\bbeta_m)$. Without loss of generality, in the following we first update $\brho$ and then update $\B{B}$, since the updating order does not affect the algorithmic convergence. Here it is not necessary to fully solve \eqref{eq:obj_mstep}; it suffices to have only one cycle of updates to lower the objective value of \eqref{eq:obj_mstep}, which is why the algorithm as a whole becomes a generalized Expectation-Maximization algorithm.

With $\B{B}$ held fixed at $\B{B}^{(t)}$, minimizing \eqref{eq:obj_mstep} leads to a closed-form update for $\brho$ as 
  \[
    \rho^{(t+1)}_j =
    \frac{\langle \tilde \by, \tilde \bX (\bbeta^{(t)}_0 + \bbeta^{(t)}_j) \rangle +
    \sqrt{\langle \tilde \by, \tilde \bX (\bbeta^{(t)}_0 + \bbeta^{(t)}_j) \rangle^2 + 4 \|\tilde \by\|^2 n_j}}
      {2\|\tilde \by\|^2} \quad (j  = 1, \ldots, m).
  \]
  where $\langle \cdot, \cdot \rangle$ denotes the inner product in Euclidean space, and $(\tilde y_i, \tilde \bx_i) = \sqrt{\widehat{p}_{ij}}(y_i, \bx_i)$ for $i = 1, \ldots, n$, 
  and $n_j = \sum_{i = 1}^n \widehat{p}_{ij}$ for $j = 1, \ldots, m$. 
  When the assumption of equal variance holds, i.e., $\rho_1 = \cdots = \rho_m$, the update becomes
\[
  \brho^{(t+1)} = \frac{\langle \by, \tilde \bmu^{(t)} \rangle +
    \sqrt{\langle \by, \tilde \bmu^{(t)} \rangle^2 + 4 \|\by\|^2 n}}
      {2\|\by\|^2} \times \1_m,
\]
where $\tilde \bmu^{(t)} = (\sum_{j=1}^m\widehat p_{1j}^{(t)} \bx_1\trans(\bbeta^{(t)}_0 + \bbeta^{(t)}_j), \ldots, \sum_{j=1}^m\widehat p_{nj}^{(t)} \bx_n\trans(\bbeta^{(t)}_0 + \bbeta^{(t)}_j))\trans$ 
and $\1_m$ is the one vector of dimension $m$.

With $\brho$ updated, it remains to solve a $\ell_1$ penalized weighted least squares with linear constraints with respect to $\B{B}$,
\begin{align}
  \min\left\{f(\bbeta_0, \bbeta_1, \ldots, \bbeta_m; \brho, \lambda)\right\}, \
  \mbox{ subject to } \sum_{j=1}^m\bbeta_{j}=\0 \label{eq:mstep-combeta}
\end{align}
where
\[
  f(\bbeta_0, \bbeta_1, \ldots, \bbeta_m; \brho, \lambda) \equiv
  \frac{1}{2} \sum_{i=1}^{n} \sum_{j=1}^{m} \widehat{p}_{ij}(\rho_jy_i-\bx_i\trans\bbeta_0-\bx_i\trans\bbeta_{j})^2 +
  n\lambda\sum_{k=1}^{p}\calP(\widetilde{\bbeta}_{k}),
\]
in which the term $n\lambda\sum_{k=1}^{p}\calP(\widetilde{\bbeta}_{k})$ is re-expressed as  $n\lambda\sum_{j=0}^{m}\calP(\bbeta_j)$ with the proposed adaptive lasso penalty form in \eqref{eq:penalty} of the main paper.

We propose a Bregman coordinate descent algorithm \citep{Goldstein2009} to solve \eqref{eq:mstep-combeta}, which performs the following updates iteratively until convergence. At the $[s+1]{\mbox{th}}$ iteration of the Bregman coordinate descent algorithm, $\B{B}$ is updated by solving an unconstrained $\ell_1$ penalized least squares
\begin{align}
  \B{B}^{[s+1]} = \arg\min_{\B{B}}\left\{
  f(\bbeta_0, \bbeta_1, \ldots, \bbeta_m;\brho^{(t+1)},
  \lambda) + 
  \frac{\mu}{2} \| \sum_{j=1}^m \bbeta_j - \bu^{[s]}\|^2
  \right\}, \label{eq:bcda}
\end{align}
using coordinate descent. Here the augmented term in the objective function is to penalize the violation of the linear constraints, in which $\bu^{[s]}$ is the discrepancy measure and $\mu > 0$ is the Bregman parameter. The $\bu^{[s]}$ is then updated as 
\[
\bu^{[s+1]} = \bu^{[s]} - \sum_{j=1}^m \bbeta_j^{[s + 1]}.
\]
The choice of $\mu > 0$ does not impact on the convergence, and it may be
increased along the iterations to improve the speed. Nonetheless, this is not essential and we simply fix $\mu = 1$ in all our numerical studies. The detailed BCDA algorithm is given in Algorithm~\ref{alg:beta}, with the derivations provided in Section~\ref{app:derive}.

\subsection{Derivations of the Bregman Coordinate Descent Algorithm}
\label{app:derive}

We provide the derivations of the coordinate-wise updating formulas in Algorithm \ref{alg:beta}.

By \eqref{eq:bcda}, at the $[s + 1]{\mbox{th}}$ iteration of the algorithm, the $k{\mbox{th}}$ entry of the common-effects vector $\bbeta_0$, denoted as $\beta_{0k}$, for $k = 1, \ldots, p$, is updated by minimizing
\begin{align*}
  & h_{0k}(\beta_{0k};\bbeta_{0, - k}, \bbeta_1, \ldots, \bbeta_m, \brho, \lambda) \\
  & = \frac{1}{2} \sum_{i=1}^{n}\sum_{j=1}^{m}\hat{p}_{ij}(\rho_jy_i-\bx_i\trans\bbeta_{j} - \bx_i\trans\bbeta_0)^2 +
    n \lambda w_{0k} |\beta_{0k}| \\
  & = \frac{1}{2} \sum_{i=1}^{n}\sum_{j=1}^{m}\hat{p}_{ij}(\rho_jy_i-\bx_i\trans\bbeta_{j} -
    \bx_{i,-k}\trans\bbeta_{0, - k} - x_{ik} \beta_{0k})^2
    + n \lambda w_{0k} |\beta_{0k}| \\
  & = \frac{1}{2}\left\{  \beta_{0k}^2 \sum_{i=1}^{n} x_{ik}^2
    - 2 \beta_{0k} [\sum_{i=1}^{n} \sum_{j=1}^{m}\hat{p}_{ij} x_{ik}(\rho_jy_i-\bx_i\trans\bbeta_{j} -
    \bx_{i,-k}\trans\bbeta_{0, - k})]
    \right\}
    + n \lambda w_{0k} |\beta_{0k}| + C,
\end{align*}
where $\bx_{i,-k}$ stands for the vector of all except $k{\mbox{th}}$ element in $\bx_i$ and the same applies to $\bbeta_{0, - k}$, and $C$ denotes a constant unrelated to $\beta_{0k}$. The minimizer is given by
\[
  \beta^{[s+1]}_{0k} = \frac{\calT(\sum_{i=1}^n \sum_{j=1}^m\hat{p}_{ij} x_{ik} (\rho_jy_i - \bx_i\trans\bbeta_{j} - \bx_{i, -k} \trans\bbeta_{0, -k}), n \lambda w_{0k})}
  {\sum_{i=1}^n x_{ik}^2},
\]
where $\calT(x, t) = \mbox{sgn}(x)(|x| - t)_{+}$ is the soft-thresholding operator.

The $k{\mbox{th}}$ element of the $j{\mbox{th}}$ cluster-effects term $\bbeta_j$, denoted as $\beta_{jk}$, for $j = 1, \ldots, m$ and $k = 1, \ldots, p$, is updated by minimizing
\begin{align*}
  & h_{jk}(\beta_{jk}; \bu, \bbeta_0, \bbeta_1, \ldots, \bbeta_{j, - k}, \ldots, \bbeta_m, \brho, \lambda) \\
  & = \frac{1}{2} \sum_{i=1}^{n} \hat{p}_{ij}(\rho_jy_i-\bx_i\trans\bbeta_0-\bx_i\trans\bbeta_{j})^2 +
    n\lambda w_{jk} |\beta_{jk}|
    + \frac{\mu}{2} (\sum_{q = 1}^m \beta_{qk} - u_k)^2 \\
  & = \frac{1}{2} \sum_{i=1}^{n} \hat{p}_{ij}(\rho_jy_i- \bx_i\trans\bbeta_0 -
    \bx_{i, -k}\trans\bbeta_{j, -k} - x_{ik} \beta_{jk})^2 +
    n\lambda w_{jk} |\beta_{jk}|
    + \frac{\mu}{2} (\beta_{jk} + \sum_{q = 1, q \ne j}^m \beta_{qk} - u_k)^2 \\
  & = \frac{1}{2} \{\beta^2_{jk} (\sum_{i=1}^{n} \hat{p}_{ij} x^2_{ik} +
    \mu)  - 2 \beta_{jk} [\sum_{i=1}^{n} \hat{p}_{ij} x_{ik} (\rho_jy_i- \bx_i\trans\bbeta_0 -
    \bx_{i, -k}\trans\bbeta_{j, -k}) - \mu (\sum_{q = 1, q \ne j}^m \beta_{qk} - u_k)]\} \\
  & + n\lambda w_{jk} |\beta_{jk}| + C,
\end{align*}
where $C$ denotes a constant unrelated to $\beta_{jk}$. The minimizer is given by
\[
  \beta^{[s+1]}_{jk} = \frac{\calT(
    \sum_{i=1}^{n} \hat{p}_{ij} x_{ik} (\rho_jy_i- \bx_i\trans\bbeta_0 -
    \bx_{i, -k}\trans\bbeta_{j, -k}) - \mu (\sum_{q = 1, q \ne j}^m
    \beta_{qk} - u_k), n \lambda w_{jk})}
  {\sum_{i=1}^{n} \hat{p}_{ij} x^2_{ik} + \mu}.
\]

\subsection{Tuning}\label{app:tuning}

The tuning parameters to be selected include the number of component $m$ and the penalty parameter $\lambda$. The adaptive estimation of the model is not really sensitive to the choice of $\gamma > 0$ in constructing the weights, so we have fixed $\gamma = 1$, which leads to satisfactory performance in our numerical studies.

The technique of cross validation certainly can be applied here, but it is computationally intensive and tends to over-select non-zero components. We propose to minimize a {\red Bayesian information criterion (BIC)},
\[
  \mbox{BIC}(m,\lambda) = - 2 \ell(\widehat \btheta_{m, \lambda}) + \log(n) \mbox{df}(m,\lambda),
\]
where $\ell(\widehat \btheta_{m, \lambda})$ is the log-likelihood evaluated at the estimated parameters, and $\mbox{df}(m,\lambda)$ denotes the degrees of freedom of the estimated model. Then, following the work by \citet{tibshirani2011gl} on the degrees of freedom of generalized lasso, we estimate $\mbox{df}(m,\lambda)$ for our problem as
\[
\widehat{df}(m, \lambda) =  2m - 1 + \mbox{dim}(\mbox{null}(\bD_{-\calA})),
\]
where $\bD = \I_p \otimes \A$ is the specified penalty matrix from generalized lasso criterion \eqref{mixreg:commplog-varsel2} of the main paper, $\calA = \{i: \{\bD\widehat\bphi_{m, \lambda}\}_i \ne 0\}$ is the index set of the non-zero entries of vector $\bD\widehat\bphi_{m, \lambda}$, $\widehat\bphi_{m, \lambda}$ is the estimated re-scaled coefficient vector $\bphi$ defined in \eqref{eq:mixreg2} of the main paper, $\bD_{-\calA}$ is the sub-matrix of $\bD$ by excluding the rows of indices in $\cal{A}$, and $\mbox{dim}(\mbox{null}(\bD_{-\calA}))$ denotes the dimension of the null space of $\bD_{-\calA}$. {\revcolor{red}With the special structure of the penalty matrix in our proposed criterion, it turns out that $\mbox{dim}(\mbox{null}(\bD_{-\calA}))$ equals to the number of nonzero entries in $\B{B}$ minus the number of predictors with heterogeneous effects. This simplified interpretation can be understood as follows. For a predictor with common effect, it only contributes one to the degrees of freedom, which is the same as the number of non-zero entries in its corresponding row in $\B{B}$. For a predictor with heteogeneous effects, its contribution to the degrees of freedom should equals to the number of non-zero entries in its corresponding row in $\B{B}$ minus 1, where the discount is due to the linear constraint that the effect terms sum up to zero.}

{\revcolor{red}
We use grid search to select $m$ and $\lambda$ by minimizing BIC. Denote by $\X_k$ the $k$th column of $\X$. Following \citet{stadler2010}, the candidate sequence of the number of components usually ranges from one to a sufficiently large number; the candidate values for $\lambda$ are generated as a sequence between 0 and $\lambda_{\mbox{max}}$ equally spaced on the log-scale, where
\[
    \lambda_{\mbox{max}} = \max_{k=1,\ldots,p}\left|\frac{\langle \by, \X_k \rangle}{\sqrt{n}\|\by\|}\right|.
\]
This ensures the coverage of a spectrum of candidate models with different sparsity levels.}

{
\vspace{6pt}
\linespread{1.4}
\begin{algorithm}
  \caption{Bregman Coordinate Descent Algorithm for solving \eqref{eq:mstep-combeta}}
  \label{alg:beta}
  \begin{algorithmic}
    \STATE Initialization: $\bbeta_j^{[0]} = \bbeta_j^{(t)}$, $j = 0, 1,\ldots, m$.  
    $\bu^{[s]} = (0, \ldots, 0)\trans_{p\times1}$. Choose $\mu > 0$.\\
    Set $s\gets0$.  
    \REPEAT
    \STATE
    $\bbeta_{j}^{[s+1]} \leftarrow \bbeta_{j}^{[s]}$, $j = 0, 1,\ldots, m$.
    \FOR{$j=0, 1,\ldots,m$}
    \FOR{$k = 1,\dots,p$}
    \IF{$j = 0$}
    \STATE
    {\fontsize{10pt}{10pt}
    \begin{align*}
      \beta_{jk}^{[s+1]} \gets
                           \frac{\calT(\sum_{i=1}^n \sum_{q=1}^m\widehat{p}_{iq} x_{ik} (\rho^{(t+1)}_qy_i - \bx_i\trans\bbeta^{[s+1]}_{q} - \bx_{i, -k} \trans\bbeta_{j, -k}^{[s+1]}), n \lambda w_{jk})}
                           {\sum_{i=1}^n x_{ik}^2}
    \end{align*}
    }
    \ELSE
    \STATE
    {\fontsize{10pt}{10pt}
    \begin{align*}
    \beta_{jk}^{[s+1]} \gets
        \frac{\calT(
          \sum_{i=1}^{n} \widehat{p}_{ij} x_{ik} (\rho^{(t + 1)}_jy_i- \bx_i\trans\bbeta^{(t + 1)}_0 -      
          \bx_{i, -k}\trans\bbeta_{j, -k}^{[s+1]}) - \mu (\sum_{q = 1, q \ne j}^m
          \beta_{qk} - u_k^{[s]}), n \lambda w_{jk})}
        {\sum_{i=1}^{n} \widehat{p}_{ij} x^2_{ik} + \mu },
    \end{align*}
    }
    \ENDIF
    \STATE \textbf{end if}
    \STATE
    where $\calT(x, t) = \mbox{sgn}(x)(|x| - t)_{+}$ is the soft-thresholding operator, and $\bx_{i, -k}$ stands for vector of all elements except $k{\mbox{th}}$ in $\bx_{i}$ and the same applies to $\bbeta_{j, -k}$.    
    \ENDFOR
    \STATE \textbf{end for}
    \ENDFOR
    \STATE \textbf{end for}
    \STATE
    $\bu^{[s+1]} = \bu^{[s]} - \sum_{j=1}^m \bbeta_j^{[s + 1]}$.\\
    $s \leftarrow s+1$.
    \UNTIL 
    $\|\bbeta^{[s]}-\bbeta^{[s-1]}\|/\|\bbeta^{[s-1]}\| \leq 10^{-6}$.
    \RETURN $\bbeta_{j}^{(t+1)} = \bbeta_{j}^{[s]}$, $j=0, 1,\ldots,m$.
  \end{algorithmic}
\end{algorithm}
}

\clearpage
\section{Additional Simulation Results}\label{app:sim}

\subsection{Results for $p \in \{30, 120\}$}
We present the detailed simulation results of the simulation setups described in the main paper for $p\in\{30, 120\}$ in Table~\ref{tab:n=30} and~\ref{tab:n=120}, respectively.

{
\def\baselinestretch{1}\normalsize
\begin{table}[h]
  \centering
  \caption{ Comparison of mean squared error of estimation, variable
    selection and heterogeneity pursuit performance of four methods,
    Mix-L, Mix-AL, Mix-HP-L and Mix-HP-L, under settings of the main paper
    with {\revcolor{red}$n = 200$, $p = 30$, and $\bSigma = \bI_p$}. The layout of the table is the same as in
    Table~\ref{tab:n=60} of the main paper. The MSE values are scaled by multiplying 100, and the FPR, FHR, TPR values are reported in percentage.} \label{tab:n=30}
  \begin{tabular}{llrrrrrr}
    \toprule
    & & \multicolumn{3}{c}{MSE} & \multicolumn{3}{c}{RATE}\\
    \cmidrule(lr){3-5} \cmidrule(lr){6-8}
    $\mbox{SNR}$ & Method & $\bb$ & $\bsigma^2$ & $\bpi$ & FPR & FHR & TPR \\
    \midrule
    200 & Mix-L & 0.10 (0.05) & 8.54 (3.20) & 0.18 (0.16) & 58.8 & 100.0 & 100.0 \\
    & Mix-AL & 0.02 (0.01) & 0.04 (0.04) & 0.12 (0.10) & 12.4 & 100.0 & 100.0 \\
    & Mix-HP-L & 0.02 (0.01) & 2.16 (1.09) & 0.14 (0.12) & 66.3 & 25.9 & 100.0 \\
    & Mix-HP-AL & 0.01 (0.00) & 0.05 (0.04) & 0.12 (0.10) & 7.5 & 1.4 & 100.0 \\
    \midrule
    100 & Mix-L & 0.20 (0.10) & 10.31 (3.97) & 0.18 (0.15) & 57.8 & 100.0 & 100.0 \\
    & Mix-AL & 0.04 (0.01) & 0.10 (0.10) & 0.11 (0.09) & 18.3 & 100.0 & 100.0 \\
    & Mix-HP-L & 0.05 (0.02) & 3.58 (1.79) & 0.13 (0.11) & 58.6 & 18.7 & 100.0 \\
    & Mix-HP-AL & 0.02 (0.01) & 0.05 (0.04) & 0.11 (0.10) & 8.1 & 1.1 & 100.0 \\
    \midrule
    50 & Mix-L & 0.51 (0.29) & 13.79 (6.28) & 0.24 (0.20) & 57.0 & 100.0 & 100.0 \\
    & Mix-AL & 0.11 (0.04) & 0.18 (0.15) & 0.13 (0.11) & 25.1 & 100.0 & 100.0 \\
    & Mix-HP-L & 0.12 (0.04) & 5.00 (2.64) & 0.16 (0.13) & 54.7 & 16.1 & 100.0 \\
    & Mix-HP-AL & 0.04 (0.02) & 0.07 (0.07) & 0.14 (0.12) & 7.6 & 1.1 & 100.0 \\
    \midrule
    25 & Mix-L & 1.13 (0.62) & 17.20 (8.87) & 0.39 (0.35) & 51.6 & 100.0 & 100.0 \\
    & Mix-AL & 0.37 (0.17) & 0.48 (0.40) & 0.15 (0.13) & 33.7 & 100.0 & 100.0 \\
    & Mix-HP-L & 0.26 (0.09) & 7.80 (4.13) & 0.20 (0.17) & 49.7 & 12.0 & 100.0 \\
    & Mix-HP-AL & 0.08 (0.04) & 0.15 (0.15) & 0.16 (0.13) & 7.3 & 1.1 & 100.0 \\
    \midrule
    12.5 & Mix-L & 7.18 (5.97) & 45.62 (43.40) & 4.78 (6.10) & 32.2 & 100.0 & 100.0 \\
    & Mix-AL & 5.02 (5.63) & 10.64 (17.96) & 1.51 (2.47) & 25.5 & 100.0 & 100.0 \\
    & Mix-HP-L & 0.61 (0.24) & 10.30 (6.06) & 0.24 (0.20) & 50.5 & 10.2 & 100.0 \\
    & Mix-HP-AL & 0.20 (0.10) & 0.31 (0.31) & 0.17 (0.14) & 9.3 & 1.7 & 100.0 \\
    \bottomrule
  \end{tabular}
\end{table}

\begin{table}[h]
  \centering
  \caption{Comparison of mean squared error of estimation, variable
    selection and heterogeneity pursuit performance of four methods,
    Mix-L, Mix-AL, Mix-HP-L and Mix-HP-L, under settings of the main paper
    with {\revcolor{red}$n = 200$, $p = 120$, and $\bSigma = \bI_p$}. The layout of the table is the same as in
    Table~\ref{tab:n=60} of main paper. The MSE values are scaled by multiplying 100, and the FPR, FHR, TPR values are reported in percentage.} \label{tab:n=120}
\begin{tabular}{llrrrrrr}
  \toprule
  & & \multicolumn{3}{c}{MSE} & \multicolumn{3}{c}{RATE}\\
  \cmidrule(lr){3-5} \cmidrule(lr){6-8}
  $\mbox{SNR}$ & Method & $\bb$ & $\bsigma^2$ & $\bpi$ & FPR & FHR & TPR \\
  \midrule
  200 & Mix-L & 0.09 (0.09) & 12.67 (9.65) & 0.23 (0.24) & 32.4 & 100.0 & 100.0 \\
  & Mix-AL & 0.02 (0.05) & 0.24 (0.56) & 0.15 (0.17) & 6.3 & 100.0 & 100.0 \\
  & Mix-HP-L & 0.02 (0.01) & 8.37 (2.87) & 0.13 (0.11) & 24.9 & 2.7 & 100.0 \\
  & Mix-HP-AL & 0.00 (0.00) & 0.06 (0.04) & 0.10 (0.08) & 1.8 & 0.1 & 100.0 \\
  \midrule
  100 & Mix-L & 2.87 (1.70) & 161.80 (118.47) & 6.20 (6.55) & 9.5 & 100.0 & 40.0 \\
  & Mix-AL & 2.65 (1.72) & 146.34 (126.12) & 3.75 (4.47) & 3.5 & 100.0 & 40.0 \\
  & Mix-HP-L & 0.04 (0.01) & 11.04 (3.98) & 0.15 (0.11) & 25.0 & 2.7 & 100.0 \\
  & Mix-HP-AL & 0.00 (0.00) & 0.07 (0.07) & 0.11 (0.09) & 2.0 & 0.0 & 100.0 \\
  \midrule
  50 & Mix-L & 3.91 (0.28) & 191.72 (88.52) & 8.18 (6.65) & 1.2 & 100.0 & 20.0 \\
  & Mix-AL & 3.61 (0.68) & 173.28 (93.38) & 5.40 (5.47) & 0.7 & 100.0 & 20.0 \\
  & Mix-HP-L & 0.09 (0.03) & 13.70 (5.23) & 0.18 (0.15) & 27.5 & 3.0 & 100.0 \\
  & Mix-HP-AL & 0.01 (0.00) & 0.14 (0.15) & 0.12 (0.10) & 2.5 & 0.1 & 100.0 \\
  \midrule
  25 & Mix-L & 3.96 (0.26) & 170.60 (82.39) & 7.99 (6.36) & 0.6 & 100.0 & 10.0 \\
  & Mix-AL & 3.74 (0.58) & 160.99 (90.05) & 5.28 (4.98) & 0.5 & 100.0 & 10.0 \\
  & Mix-HP-L & 1.37 (1.75) & 74.93 (96.44) & 0.40 (0.47) & 21.5 & 2.5 & 100.0 \\
  & Mix-HP-AL & 1.24 (1.79) & 55.32 (88.47) & 0.28 (0.34) & 2.3 & 0.1 & 100.0 \\
  \midrule
  12.5 & Mix-L & 4.00 (0.22) & 155.17 (81.24) & 7.13 (6.62) & 0.4 & 100.0 & 10.0 \\
  & Mix-AL & 3.90 (0.42) & 146.84 (85.65) & 4.44 (4.64) & 0.3 & 100.0 & 10.0 \\
  & Mix-HP-L & 3.93 (0.26) & 184.75 (66.99) & 2.41 (2.26) & 2.2 & 0.3 & 20.0 \\
  & Mix-HP-AL & 3.87 (0.33) & 156.13 (71.01) & 2.24 (2.14) & 0.7 & 0.0 & 20.0 \\
  \bottomrule
\end{tabular}
\end{table}
}

\clearpage
\subsection{Correlation Structure Among Predictors}\label{app:add_sim_cor}

We also experiment with settings where the predictors are generated from multivariate normal 
distribution with correlation structure to make the simulation setting more realistic. Here we 
consider the same correlation structure of \citet{khalili2012variable}. {\revcolor{red}
That is, the data on the predictors, $\x_i \in \mathbb{R}^p$ for $ i = 1, \ldots, n$, are generated independently from multivariate normal distribution with mean $\mathbf{0}$ and covariance matrix $\bSigma$, where the $(i,j)$'s entry of $\bSigma$ is given by $\sigma_{ij} = 0.5^{|i-j|}$.} The other settings remain the same as the simulation settings of the main paper. 
The simulation results are presented in Table~\ref{tab:app_2_n=30}, \ref{tab:app_2_n=60} and \ref{tab:app_2_n=120}.

{
\def\baselinestretch{1}\normalsize
\begin{table}[htp]
  \centering
  \caption{Comparison of mean squared error of estimation, variable
    selection and heterogeneity pursuit performance of four methods,
    Mix-L, Mix-AL, Mix-HP-L and Mix-HP-L, under settings where correlations present among 
    predictors with
    {\revcolor{red}$n = 200$, $p = 30$, and $\bSigma = (\sigma_{ij})_{p\times p}$ where $\sigma_{ij} = 0.5^{|i-j|}$}. The layout of the table is the same as in
    Table~\ref{tab:n=60} of the main paper. The MSE values are scaled by multiplying 100, and the FPR, FHR, TPR values are reported in percentage.} \label{tab:app_2_n=30}
  \begin{tabular}{llrrrrrr}
    \toprule
    & & \multicolumn{3}{c}{MSE} & \multicolumn{2}{c}{RATE}\\
    \cmidrule(lr){3-5} \cmidrule(lr){6-7}
    $\mbox{SNR}$ & Method & $\bb$ & $\bsigma^2$ & $\bpi$ & FPR & TPR \\
    \midrule
    250 & Mix-L & 1.53 (1.73) & 10.47 (5.75) & 2.79 (1.91) & 61.2 & 100.0 & 100.0 \\ 
    & Mix-AL & 0.03 (0.01) & 0.06 (0.05) & 1.42 (0.56) & 16.2 & 100.0 & 100.0 \\ 
    & Mix-HP-L & 0.03 (0.01) & 1.19 (0.69) & 1.45 (0.58) & 68.0 & 28.5 & 100.0 \\ 
    & Mix-HP-AL & 0.01 (0.00) & 0.06 (0.04) & 1.45 (0.57) & 6.1 & 1.0 & 100.0 \\ 
    \midrule
    125 & Mix-L & 2.32 (1.80) & 12.57 (6.58) & 3.49 (2.04) & 60.3 & 100.0 & 100.0 \\ 
    & Mix-AL & 0.07 (0.03) & 0.10 (0.08) & 1.39 (0.57) & 24.4 & 100.0 & 100.0 \\ 
    & Mix-HP-L & 0.06 (0.02) & 1.82 (1.15) & 1.47 (0.62) & 64.0 & 26.8 & 100.0 \\ 
    & Mix-HP-AL & 0.02 (0.01) & 0.07 (0.06) & 1.46 (0.59) & 7.1 & 1.1 & 100.0 \\ 
    \midrule
    62.5 & Mix-L & 3.42 (1.54) & 16.44 (9.72) & 4.59 (1.81) & 54.8 & 100.0 & 100.0 \\ 
    & Mix-AL & 1.38 (1.78) & 1.77 (3.25) & 2.45 (1.93) & 27.4 & 100.0 & 100.0 \\ 
    & Mix-HP-L & 0.14 (0.05) & 2.76 (1.82) & 1.55 (0.67) & 59.5 & 23.4 & 100.0 \\ 
    & Mix-HP-AL & 0.04 (0.02) & 0.09 (0.10) & 1.49 (0.63) & 7.1 & 1.1 & 100.0 \\
    \midrule
    31.3 & Mix-L & 4.52 (0.41) & 18.65 (11.00) & 5.65 (0.16) & 47.9 & 100.0 & 100.0 \\ 
    & Mix-AL & 4.37 (0.63) & 6.88 (5.39) & 5.40 (0.57) & 21.2 & 100.0 & 100.0 \\ 
    & Mix-HP-L & 0.37 (0.14) & 3.81 (2.54) & 1.59 (0.70) & 57.4 & 20.7 & 100.0 \\ 
    & Mix-HP-AL & 0.11 (0.05) & 0.19 (0.19) & 1.49 (0.65) & 7.2 & 1.2 & 100.0 \\ 
    \midrule
    15.6 & Mix-L & 4.94 (0.42) & 20.33 (10.65) & 5.68 (0.20) & 39.0 & 100.0 & 90.0 \\ 
    & Mix-AL & 4.83 (0.48) & 10.70 (7.15) & 5.50 (0.35) & 15.6 & 100.0 & 90.0 \\ 
    & Mix-HP-L & 0.77 (0.27) & 4.36 (3.22) & 1.75 (0.81) & 56.6 & 19.4 & 100.0 \\ 
    & Mix-HP-AL & 0.28 (0.14) & 0.42 (0.42) & 1.51 (0.70) & 9.7 & 2.0 & 100.0 \\ 
    \bottomrule
  \end{tabular}
\end{table}

\begin{table}[htp]
  \centering
  \caption{Comparison of mean squared error of estimation, variable
    selection and heterogeneity pursuit performance of four methods,
    Mix-L, Mix-AL, Mix-HP-L and Mix-HP-L, under settings where correlations present among 
    predictors with
    {\revcolor{red}$n = 200$, $p = 60$, and $\bSigma = (\sigma_{ij})_{p\times p}$ where $\sigma_{ij} = 0.5^{|i-j|}$}. The layout of the table is the same as in
    Table~\ref{tab:n=60} of the main paper. The MSE values are scaled by multiplying 100, and the FPR, FHR, TPR values are reported in percentage.} \label{tab:app_2_n=60}
  \begin{tabular}{llrrrrrr}
    \toprule
    & & \multicolumn{3}{c}{MSE} & \multicolumn{2}{c}{RATE}\\
    \cmidrule(lr){3-5} \cmidrule(lr){6-7}
    $\mbox{SNR}$ & Method & $\bb$ & $\bsigma^2$ & $\bpi$ & FPR & TPR \\
    \midrule
    250 & Mix-L & 0.73 (0.92) & 8.62 (10.11) & 2.68 (1.98) & 61.4 & 100.0 & 100.0 \\ 
   & Mix-AL & 0.02 (0.02) & 0.19 (0.23) & 1.36 (0.62) & 14.6 & 100.0 & 100.0 \\ 
   & Mix-HP-L & 0.03 (0.01) & 3.33 (1.75) & 1.41 (0.58) & 45.2 & 12.1 & 100.0 \\ 
   & Mix-HP-AL & 0.00 (0.00) & 0.08 (0.05) & 1.42 (0.57) & 3.6 & 0.4 & 100.0 \\ 
   \midrule
   125 & Mix-L & 2.14 (0.39) & 26.65 (18.13) & 5.19 (1.41) & 40.0 & 100.0 & 100.0 \\ 
   & Mix-AL & 1.62 (0.91) & 16.22 (15.47) & 4.34 (1.98) & 13.7 & 100.0 & 100.0 \\ 
   & Mix-HP-L & 0.06 (0.02) & 4.72 (2.44) & 1.60 (0.68) & 42.7 & 12.1 & 100.0 \\ 
   & Mix-HP-AL & 0.01 (0.00) & 0.09 (0.08) & 1.51 (0.62) & 3.3 & 0.4 & 100.0 \\ 
   \midrule
   62.5 & Mix-L & 2.35 (0.19) & 34.57 (14.43) & 5.71 (0.21) & 26.4 & 100.0 & 90.0 \\ 
   & Mix-AL & 2.27 (0.15) & 18.96 (12.12) & 5.63 (0.10) & 10.0 & 100.0 & 90.0 \\ 
   & Mix-HP-L & 0.15 (0.06) & 5.90 (3.23) & 1.66 (0.73) & 43.6 & 13.4 & 100.0 \\ 
   & Mix-HP-AL & 0.02 (0.01) & 0.14 (0.15) & 1.46 (0.63) & 3.8 & 0.5 & 100.0 \\
   \midrule
   31.3 & Mix-L & 2.55 (0.27) & 38.44 (15.45) & 5.82 (0.35) & 21.1 & 100.0 & 90.0 \\ 
   & Mix-AL & 2.36 (0.18) & 16.80 (11.18) & 5.64 (0.13) & 8.6 & 100.0 & 90.0 \\ 
   & Mix-HP-L & 0.34 (0.13) & 6.63 (4.32) & 1.75 (0.82) & 46.0 & 14.6 & 100.0 \\ 
   & Mix-HP-AL & 0.06 (0.03) & 0.22 (0.24) & 1.44 (0.69) & 5.2 & 1.1 & 100.0 \\ 
   \midrule
   15.6 & Mix-L & 2.92 (0.45) & 35.75 (15.89) & 5.78 (0.34) & 18.9 & 100.0 & 90.0 \\ 
   & Mix-AL & 2.67 (0.33) & 16.37 (11.16) & 5.63 (0.10) & 8.6 & 100.0 & 90.0 \\ 
   & Mix-HP-L & 0.98 (0.74) & 7.22 (4.94) & 1.93 (0.97) & 43.2 & 10.3 & 100.0 \\ 
   & Mix-HP-AL & 0.61 (0.74) & 0.96 (1.23) & 1.53 (0.80) & 7.0 & 2.1 & 100.0 \\ 
    \bottomrule
  \end{tabular}
\end{table}

\begin{table}[htp]
  \centering
  \caption{Comparison of mean squared error of estimation, variable
    selection and heterogeneity pursuit performance of four methods,
    Mix-L, Mix-AL, Mix-HP-L and Mix-HP-L, under settings where correlations present among 
    predictors with
    {\revcolor{red}$n = 200$, $p = 120$, and $\bSigma = (\sigma_{ij})_{p\times p}$ where $\sigma_{ij} = 0.5^{|i-j|}$}. The layout of the table is the same as in
    Table~\ref{tab:n=60} of the main paper. The MSE values are scaled by multiplying 100, and the FPR, FHR, TPR values are reported in percentage.} \label{tab:app_2_n=120}
  \begin{tabular}{llrrrrrr}
    \toprule
    & & \multicolumn{3}{c}{MSE} & \multicolumn{2}{c}{RATE}\\
    \cmidrule(lr){3-5} \cmidrule(lr){6-7}
    $\mbox{SNR}$ & Method & $\bb$ & $\bsigma^2$ & $\bpi$ & FPR & TPR \\
    \midrule
    250 & Mix-L & 1.18 (0.08) & 39.87 (21.91) & 5.66 (0.16) & 22.2 & 100.0 & 90.0 \\ 
    & Mix-AL & 1.14 (0.09) & 20.91 (20.68) & 5.62 (0.08) & 9.6 & 100.0 & 90.0 \\ 
    & Mix-HP-L & 0.03 (0.01) & 6.04 (2.32) & 1.49 (0.63) & 28.6 & 5.7 & 100.0 \\ 
    & Mix-HP-AL & 0.00 (0.00) & 0.10 (0.07) & 1.45 (0.58) & 1.8 & 0.1 & 100.0 \\ 
    \midrule
    125 & Mix-L & 1.22 (0.10) & 45.55 (20.94) & 5.71 (0.21) & 16.7 & 100.0 & 90.0 \\ 
    & Mix-AL & 1.15 (0.06) & 23.79 (22.03) & 5.64 (0.11) & 7.5 & 100.0 & 90.0 \\ 
    & Mix-HP-L & 0.06 (0.02) & 7.17 (3.10) & 1.60 (0.72) & 30.5 & 6.8 & 100.0 \\ 
    & Mix-HP-AL & 0.01 (0.00) & 0.12 (0.10) & 1.50 (0.62) & 2.1 & 0.2 & 100.0 \\ 
    \midrule
    62.5 & Mix-L & 1.31 (0.13) & 52.03 (21.44) & 5.79 (0.30) & 13.1 & 100.0 & 90.0 \\ 
    & Mix-AL & 1.20 (0.09) & 20.78 (16.83) & 5.63 (0.10) & 6.2 & 100.0 & 90.0 \\ 
    & Mix-HP-L & 0.12 (0.04) & 7.95 (4.14) & 1.71 (0.74) & 33.0 & 9.2 & 100.0 \\ 
    & Mix-HP-AL & 0.01 (0.01) & 0.19 (0.19) & 1.48 (0.61) & 2.2 & 0.2 & 100.0 \\ 
    \midrule
    31.3 & Mix-L & 1.50 (0.25) & 58.17 (22.02) & 5.89 (0.37) & 10.1 & 100.0 & 90.0 \\ 
    & Mix-AL & 1.24 (0.10) & 33.30 (29.22) & 5.66 (0.14) & 4.5 & 100.0 & 90.0 \\ 
    & Mix-HP-L & 0.23 (0.09) & 7.59 (4.81) & 1.64 (0.74) & 34.0 & 9.3 & 100.0 \\ 
    & Mix-HP-AL & 0.03 (0.01) & 0.14 (0.15) & 1.35 (0.62) & 3.2 & 0.9 & 100.0 \\ 
    \midrule
    15.6 & Mix-L & 1.51 (0.13) & 34.49 (10.23) & 6.06 (0.49) & 14.2 & 100.0 & 90.0 \\ 
    & Mix-AL & 1.39 (0.20) & 34.95 (30.42) & 5.49 (0.36) & 2.1 & 100.0 & 90.0 \\ 
    & Mix-HP-L & 0.52 (0.46) & 11.26 (13.90) & 2.50 (1.33) & 26.1 & 0.0 & 100.0 \\ 
    & Mix-HP-AL & 0.28 (0.32) & 2.10 (2.58) & 1.98 (1.09) & 1.5 & 0.0 & 100.0 \\ 
    \bottomrule
  \end{tabular}
\end{table}
}

\clearpage
\subsection{Results for the Settings of Unequal Mixing Probabilities}\label{app:add_sim_ub}
We experiment with settings where the mixing probabilities are unequal. We set $\pi_1 = \pi_2 = 0.25$ and 
$\pi_3 = 0.5$. All the other settings are the same as those in
Section~\ref{sec:sim} of the main paper. The simulation results are presented in Table~\ref{tab:ub_n=30},~\ref{tab:ub_n=60} and~\ref{tab:ub_n=120}.

{
\def\baselinestretch{1}\normalsize
\begin{table}[h]
  \centering
  \caption{ Comparison of mean squared error of estimation, variable
    selection and heterogeneity pursuit performance of four methods,
    Mix-L, Mix-AL, Mix-HP-L and Mix-HP-L, under settings with {\revcolor{red}$n = 200$, $p =
    30$, and $\bSigma = \bI_p$}. The mixing probabilities are set as $\pi_1 = \pi_2 = 0.25$ and $\pi_3 = 0.5$. 
    The layout of the table is the same as in
    Table~\ref{tab:n=60} of the main paper. The MSE values are scaled by multiplying 100, and the FPR, FHR, TPR values are reported in percentage.} \label{tab:ub_n=30}
  \begin{tabular}{llrrrrrr}
    \toprule
    & & \multicolumn{3}{c}{MSE} & \multicolumn{3}{c}{RATE}\\
    \cmidrule(lr){3-5} \cmidrule(lr){6-8}
    $\mbox{SNR}$ & Method & $\bb$ & $\bsigma^2$ & $\bpi$ & FPR & FHR & TPR \\
    \midrule
    200 & Mix-L & 0.33 (0.32) & 10.85 (5.51) & 0.21 (0.20) & 61.0 & 100.0 & 100.0 \\ 
    & Mix-AL & 0.02 (0.01) & 0.04 (0.04) & 0.09 (0.09) & 15.2 & 100.0 & 100.0 \\ 
    & Mix-HP-L & 0.02 (0.01) & 1.94 (1.14) & 0.10 (0.10) & 68.0 & 29.6 & 100.0 \\ 
    & Mix-HP-AL & 0.01 (0.00) & 0.05 (0.04) & 0.09 (0.09) & 7.7 & 1.1 & 100.0 \\ 
    \midrule
    100 & Mix-L & 0.52 (0.39) & 13.18 (6.22) & 0.28 (0.27) & 60.4 & 100.0 & 100.0 \\ 
    & Mix-AL & 0.05 (0.02) & 0.10 (0.09) & 0.10 (0.09) & 23.1 & 100.0 & 100.0 \\ 
    & Mix-HP-L & 0.05 (0.01) & 3.07 (1.68) & 0.12 (0.11) & 60.9 & 24.1 & 100.0 \\ 
    & Mix-HP-AL & 0.02 (0.01) & 0.05 (0.05) & 0.10 (0.09) & 7.7 & 1.6 & 100.0 \\ 
    \midrule
    50 & Mix-L & 0.93 (0.57) & 16.46 (8.71) & 0.33 (0.29) & 57.7 & 100.0 & 100.0 \\ 
    & Mix-AL & 0.14 (0.06) & 0.21 (0.17) & 0.14 (0.11) & 32.2 & 100.0 & 100.0 \\ 
    & Mix-HP-L & 0.10 (0.03) & 4.18 (2.25) & 0.14 (0.12) & 56.5 & 21.1 & 100.0 \\ 
    & Mix-HP-AL & 0.03 (0.02) & 0.07 (0.07) & 0.12 (0.10) & 6.6 & 1.1 & 100.0 \\
    \midrule
    25 & Mix-L & 1.36 (0.60) & 17.35 (9.27) & 0.42 (0.34) & 57.1 & 100.0 & 100.0 \\ 
    & Mix-AL & 0.77 (0.50) & 0.88 (0.82) & 0.22 (0.19) & 33.9 & 100.0 & 100.0 \\ 
    & Mix-HP-L & 0.24 (0.08) & 5.70 (3.24) & 0.16 (0.12) & 53.6 & 17.9 & 100.0 \\ 
    & Mix-HP-AL & 0.07 (0.03) & 0.13 (0.13) & 0.13 (0.11) & 7.2 & 1.2 & 100.0 \\ 
    \bottomrule
  \end{tabular}
\end{table}

\begin{table}[h]
  \centering
  \caption{ Comparison of mean squared error of estimation, variable
    selection and heterogeneity pursuit performance of four methods,
    Mix-L, Mix-AL, Mix-HP-L and Mix-HP-L, under settings with {\revcolor{red}$n = 200$, $p =
    60$, and $\bSigma = \bI_p$}. The mixing probabilities are set as $\pi_1 = \pi_2 = 0.25$ and $\pi_3 = 0.5$. 
    The layout of the table is the same as in
    Table~\ref{tab:n=60} of the main paper. The MSE values are scaled by multiplying 100, and the FPR, FHR, TPR values are reported in percentage.} \label{tab:ub_n=60}
  \begin{tabular}{llrrrrrr}
    \toprule
    & & \multicolumn{3}{c}{MSE} & \multicolumn{3}{c}{RATE}\\
    \cmidrule(lr){3-5} \cmidrule(lr){6-8}
    $\mbox{SNR}$ & Method & $\bb$ & $\bsigma^2$ & $\bpi$ & FPR & FHR & TPR \\
    \midrule
    200 & Mix-L & 0.06 (0.03) & 5.73 (3.38) & 0.18 (0.16) & 62.9 & 100.0 & 100.0 \\ 
    & Mix-AL & 0.01 (0.01) & 0.14 (0.13) & 0.11 (0.10) & 16.6 & 100.0 & 100.0 \\ 
    & Mix-HP-L & 0.01 (0.00) & 4.23 (1.94) & 0.12 (0.10) & 42.4 & 8.6 & 100.0 \\ 
    & Mix-HP-AL & 0.00 (0.00) & 0.08 (0.05) & 0.10 (0.08) & 3.6 & 0.2 & 100.0 \\ 
    \midrule
    100 & Mix-L & 0.24 (0.20) & 10.22 (7.59) & 0.31 (0.30) & 56.4 & 100.0 & 100.0 \\ 
    & Mix-AL & 0.14 (0.17) & 0.52 (0.75) & 0.17 (0.17) & 17.9 & 100.0 & 100.0 \\ 
    & Mix-HP-L & 0.03 (0.01) & 6.18 (2.61) & 0.12 (0.11) & 38.2 & 6.9 & 100.0 \\ 
    & Mix-HP-AL & 0.01 (0.00) & 0.08 (0.07) & 0.09 (0.08) & 3.5 & 0.2 & 100.0 \\ 
    \midrule
    50 & Mix-L & 0.49 (0.32) & 17.78 (13.34) & 0.46 (0.41) & 49.7 & 100.0 & 100.0 \\ 
    & Mix-AL & 0.32 (0.23) & 2.14 (2.63) & 0.29 (0.25) & 21.1 & 100.0 & 100.0 \\ 
    & Mix-HP-L & 0.07 (0.02) & 8.57 (3.45) & 0.14 (0.12) & 34.5 & 5.8 & 100.0 \\ 
    & Mix-HP-AL & 0.02 (0.01) & 0.07 (0.08) & 0.11 (0.09) & 4.0 & 0.3 & 100.0 \\ 
    \midrule
    25 & Mix-L & 2.63 (0.35) & 57.82 (27.96) & 3.28 (0.52) & 22.3 & 100.0 & 90.0 \\ 
    & Mix-AL & 2.38 (0.34) & 30.48 (24.89) & 2.70 (1.51) & 11.4 & 100.0 & 90.0 \\ 
    & Mix-HP-L & 0.17 (0.06) & 10.88 (5.00) & 0.19 (0.16) & 35.7 & 5.9 & 100.0 \\ 
    & Mix-HP-AL & 0.04 (0.02) & 0.20 (0.21) & 0.13 (0.11) & 3.8 & 0.4 & 100.0 \\ 
    \bottomrule
  \end{tabular}
\end{table}

\begin{table}[h]
  \centering
  \caption{ Comparison of mean squared error of estimation, variable
    selection and heterogeneity pursuit performance of four methods,
    Mix-L, Mix-AL, Mix-HP-L and Mix-HP-L, under settings with {\revcolor{red}$n = 200$, $p =
    120$, and $\bSigma = \bI_p$}. The mixing probabilities are set as $\pi_1 = \pi_2 = 0.25$ and $\pi_3 = 0.5$. 
    The layout of the table is the same as in
    Table~\ref{tab:n=60} of the main paper. The MSE values are scaled by multiplying 100, and the FPR, FHR, TPR values are reported in percentage.} \label{tab:ub_n=120}
  \begin{tabular}{llrrrrrr}
    \toprule
    & & \multicolumn{3}{c}{MSE} & \multicolumn{3}{c}{RATE}\\
    \cmidrule(lr){3-5} \cmidrule(lr){6-8}
    $\mbox{SNR}$ & Method & $\bb$ & $\bsigma^2$ & $\bpi$ & FPR & FHR & TPR \\
    \midrule
    200 & Mix-L & 0.21 (0.18) & 10.20 (7.59) & 0.21 (0.19) & 45.0 & 100.0 & 100.0 \\ 
    & Mix-AL & 0.15 (0.17) & 1.37 (2.14) & 0.17 (0.14) & 14.5 & 100.0 & 100.0 \\ 
    & Mix-HP-L & 0.01 (0.00) & 7.13 (2.64) & 0.14 (0.11) & 26.1 & 5.8 & 100.0 \\ 
    & Mix-HP-AL & 0.00 (0.00) & 0.05 (0.04) & 0.12 (0.09) & 1.7 & 0.3 & 100.0 \\ 
    \midrule
    100 & Mix-L & 0.42 (0.28) & 31.74 (21.74) & 0.45 (0.36) & 34.4 & 100.0 & 100.0 \\ 
    & Mix-AL & 0.29 (0.23) & 5.34 (5.72) & 0.30 (0.21) & 12.3 & 100.0 & 100.0 \\ 
    & Mix-HP-L & 0.03 (0.01) & 9.93 (4.22) & 0.20 (0.15) & 24.8 & 5.6 & 100.0 \\ 
    & Mix-HP-AL & 0.00 (0.00) & 0.06 (0.06) & 0.16 (0.13) & 1.6 & 0.0 & 100.0 \\ 
    \midrule
    50 & Mix-L & 2.40 (0.90) & 100.46 (36.41) & 5.92 (3.33) & 8.2 & 100.0 & 90.0 \\ 
    & Mix-AL & 1.48 (0.27) & 81.84 (46.13) & 2.06 (1.66) & 3.5 & 100.0 & 90.0 \\ 
    & Mix-HP-L & 0.09 (0.03) & 10.35 (4.73) & 0.18 (0.15) & 31.6 & 8.5 & 100.0 \\ 
    & Mix-HP-AL & 0.01 (0.00) & 0.14 (0.14) & 0.11 (0.10) & 2.5 & 0.3 & 100.0 \\ 
    \midrule
    25 & Mix-L & 2.66 (0.88) & 97.19 (33.94) & 6.15 (3.33) & 6.1 & 100.0 & 80.0 \\ 
    & Mix-AL & 1.66 (0.43) & 82.97 (45.91) & 1.69 (1.60) & 3.0 & 100.0 & 80.0 \\ 
    & Mix-HP-L & 1.31 (1.50) & 73.11 (88.89) & 1.34 (1.77) & 23.3 & 6.4 & 100.0 \\ 
    & Mix-HP-AL & 1.18 (1.53) & 47.14 (70.28) & 1.44 (2.10) & 3.0 & 0.5 & 100.0 \\ 
    \bottomrule
  \end{tabular}
\end{table}
}

\clearpage
\subsection{Heterogeneous Effects on All Relevant Predictors}\label{app:add_sim}

We experiment with settings where all the relevant predictors have
heterogeneous effects, making heterogeneity pursuit not really necessary. So we expect that the proposed methods perform similarly as their counterparts which do not pursue sources of heterogeneity.

We consider the same dimensional settings as before with $p \in \{30, 60, 120\}$. In each setting, the first $p_0= 5$
predictors have effects to the response, and the effects are heterogeneous according to Definition \ref{def:2} of the main paper. Specifically, under the mixture effects
model~\eqref{eq:mixreg2} of the main paper with $m=3$ components, the sub-vectors of the
first 5 entries of the scaled coefficient vectors $\bphi_j$, denoted
as $\bphi_{j0}$, $j = 1, 2, 3$, are set as
\begin{align*}
  &\bphi_{10} = (1, -1, 0, -3, 3)\trans / \sqrt{\delta}, \quad
    \bphi_{20} = (-1, 2, -3, 3, 0)\trans / \sqrt{\delta},\\
  &\bphi_{30} = (2, 1, 3, 0, -3)\trans / \sqrt{\delta}.
\end{align*}
The $\delta$ is set to control $\mbox{SNR} = \{22.5, 45, 90, 180\}$. The true number of components $m = 3$ is assumed to be known. All the other settings are the same as those in
Section~\ref{sec:sim} of the main paper. The simulation results are presented in Table~\ref{tab:app_n=30},~\ref{tab:app_n=60} and~\ref{tab:app_n=120}.

{
\def\baselinestretch{1}\normalsize
\begin{table}[htp]
  \centering
  \caption{ Comparison of mean squared error of estimation, variable
    selection and heterogeneity pursuit performance of four methods,
    Mix-L, Mix-AL, Mix-HP-L and Mix-HP-L, under settings where all the relevant predictors have heterogeneous effects with {\revcolor{red}$n = 200$, $p =
    30$, and $\bSigma = \bI_p$}. The layout of the table is the same as in
    Table~\ref{tab:n=60} of the main paper. The MSE values are scaled by multiplying 100, and the FPR, FHR, TPR values are reported in percentage.} \label{tab:app_n=30}
  \begin{tabular}{llrrrrrr}
    \toprule
    & & \multicolumn{3}{c}{MSE} & \multicolumn{2}{c}{RATE}\\
    \cmidrule(lr){3-5} \cmidrule(lr){6-7}
    $\mbox{SNR}$ & Method & $\bb$ & $\bsigma^2$ & $\bpi$ & FPR & TPR \\
    \midrule
    180 & Mix-L & 0.02 (0.01) & 2.19 (1.67) & 0.11 (0.10) & 64.7 & 100.0 \\
    & Mix-AL & 0.01 (0.00) & 0.06 (0.04) & 0.09 (0.08) & 5.8 & 100.0 \\
    & Mix-HP-L & 0.03 (0.01) & 1.67 (1.33) & 0.11 (0.09) & 74.2 & 100.0 \\
    & Mix-HP-AL & 0.01 (0.00) & 0.05 (0.04) & 0.09 (0.08) & 7.1 & 100.0 \\
    \midrule
    90 & Mix-L & 0.05 (0.02) & 4.14 (2.51) & 0.14 (0.11) & 52.5 & 100.0 \\
    & Mix-AL & 0.02 (0.01) & 0.06 (0.05) & 0.11 (0.09) & 7.1 & 100.0 \\
    & Mix-HP-L & 0.05 (0.02) & 3.38 (2.03) & 0.13 (0.11) & 62.1 & 100.0 \\
    & Mix-HP-AL & 0.02 (0.01) & 0.05 (0.05) & 0.11 (0.09) & 8.1 & 100.0 \\
    \midrule
    45 & Mix-L & 0.11 (0.04) & 6.41 (3.71) & 0.14 (0.12) & 45.3 & 100.0 \\
    & Mix-AL & 0.04 (0.02) & 0.11 (0.09) & 0.12 (0.10) & 8.6 & 100.0 \\
    & Mix-HP-L & 0.11 (0.04) & 4.84 (2.44) & 0.14 (0.12) & 56.6 & 100.0 \\
    & Mix-HP-AL & 0.04 (0.02) & 0.07 (0.07) & 0.12 (0.10) & 8.2 & 100.0 \\
    \midrule
    22.5 & Mix-L & 0.27 (0.10) & 8.77 (4.46) & 0.16 (0.14) & 40.7 & 100.0 \\
    & Mix-AL & 0.09 (0.05) & 0.25 (0.20) & 0.12 (0.11) & 10.3 & 100.0 \\
    & Mix-HP-L & 0.26 (0.09) & 7.26 (3.98) & 0.15 (0.14) & 51.3 & 100.0 \\
    & Mix-HP-AL & 0.09 (0.04) & 0.17 (0.18) & 0.12 (0.11) & 8.2 & 100.0 \\
    \bottomrule
  \end{tabular}
\end{table}

\begin{table}[htp]
  \centering
  \caption{ Comparison of mean squared error of estimation, variable
    selection and heterogeneity pursuit performance of four methods,
    Mix-L, Mix-AL, Mix-HP-L and Mix-HP-L, under settings where all the relevant predictors have heterogeneous effects with {\revcolor{red}$n = 200$, $p =
    60$, and $\bSigma = \bI_p$}. The layout of the table is the same as in
    Table~\ref{tab:n=60} of the main paper. The MSE values are scaled by multiplying 100, and the FPR, FHR, TPR values are reported in percentage.} \label{tab:app_n=60}
  \begin{tabular}{llrrrrrr}
    \toprule
    & & \multicolumn{3}{c}{MSE} & \multicolumn{2}{c}{RATE}\\
    \cmidrule(lr){3-5} \cmidrule(lr){6-7}
    $\mbox{SNR}$ & Method & $\bb$ & $\bsigma^2$ & $\bpi$ & FPR & TPR \\
    \midrule
    180 & Mix-L & 0.02 (0.01) & 7.47 (2.65) & 0.14 (0.11) & 25.6 & 100.0 \\
    & Mix-AL & 0.00 (0.00) & 0.04 (0.04) & 0.10 (0.09) & 4.0 & 100.0 \\
    & Mix-HP-L & 0.02 (0.01) & 5.32 (1.96) & 0.13 (0.11) & 41.5 & 100.0 \\
    & Mix-HP-AL & 0.00 (0.00) & 0.06 (0.04) & 0.10 (0.09) & 3.7 & 100.0 \\
    \midrule
    90 & Mix-L & 0.05 (0.02) & 9.87 (3.63) & 0.15 (0.13) & 24.4 & 100.0 \\
    & Mix-AL & 0.01 (0.00) & 0.06 (0.05) & 0.11 (0.10) & 4.9 & 100.0 \\
    & Mix-HP-L & 0.05 (0.02) & 7.36 (3.13) & 0.14 (0.12) & 39.3 & 100.0 \\
    & Mix-HP-AL & 0.01 (0.00) & 0.07 (0.07) & 0.11 (0.10) & 4.8 & 100.0 \\
    \midrule
    45 & Mix-L & 0.11 (0.04) & 12.81 (5.03) & 0.18 (0.15) & 24.3 & 100.0 \\
    & Mix-AL & 0.02 (0.01) & 0.12 (0.10) & 0.12 (0.11) & 5.8 & 100.0 \\
    & Mix-HP-L & 0.10 (0.04) & 9.52 (4.20) & 0.16 (0.14) & 39.2 & 100.0 \\
    & Mix-HP-AL & 0.03 (0.01) & 0.09 (0.10) & 0.12 (0.10) & 5.6 & 100.0 \\
    \midrule
    22.5 & Mix-L & 0.25 (0.11) & 16.94 (7.93) & 0.20 (0.17) & 24.3 & 100.0 \\
    & Mix-AL & 0.07 (0.05) & 0.32 (0.26) & 0.14 (0.11) & 7.6 & 100.0 \\
    & Mix-HP-L & 0.23 (0.07) & 12.32 (5.81) & 0.19 (0.16) & 41.9 & 100.0 \\
    & Mix-HP-AL & 0.06 (0.03) & 0.28 (0.30) & 0.14 (0.12) & 7.9 & 100.0 \\
    \bottomrule
  \end{tabular}
\end{table}

\begin{table}[htp]
  \centering
  \caption{Comparison of mean squared error of estimation, variable
    selection and heterogeneity pursuit performance of four methods,
    Mix-L, Mix-AL, Mix-HP-L and Mix-HP-L, under settings where all the relevant predictors have heterogeneous effects with {\revcolor{red}$n = 200$, $p =
    120$, and $\bSigma = \bI_p$}. The layout of the table is the same as in
    Table~\ref{tab:n=60} of the main paper. The MSE values are scaled by multiplying 100, and the FPR, FHR, TPR values are reported in percentage.} \label{tab:app_n=120}
  \begin{tabular}{llrrrrrr}
    \toprule
    & & \multicolumn{3}{c}{MSE} & \multicolumn{2}{c}{RATE}\\
    \cmidrule(lr){3-5} \cmidrule(lr){6-7}
    $\mbox{SNR}$ & Method & $\bb$ & $\bsigma^2$ & $\bpi$ & FPR & TPR \\
    \midrule
    180 & Mix-L & 0.01 (0.01) & 10.72 (3.27) & 0.13 (0.11) & 12.7 & 100.0 \\
    & Mix-AL & 0.00 (0.00) & 0.03 (0.03) & 0.11 (0.09) & 2.7 & 100.0 \\
    & Mix-HP-L & 0.02 (0.01) & 8.52 (3.05) & 0.12 (0.10) & 26.0 & 100.0 \\
    & Mix-HP-AL & 0.00 (0.00) & 0.07 (0.05) & 0.10 (0.09) & 2.0 & 100.0 \\
    \midrule
    90 & Mix-L & 0.03 (0.01) & 15.47 (5.21) & 0.14 (0.12) & 11.8 & 100.0 \\
    & Mix-AL & 0.01 (0.00) & 0.05 (0.05) & 0.10 (0.08) & 3.0 & 100.0 \\
    & Mix-HP-L & 0.04 (0.01) & 11.48 (4.30) & 0.13 (0.11) & 25.9 & 100.0 \\
    & Mix-HP-AL & 0.00 (0.00) & 0.09 (0.09) & 0.10 (0.08) & 2.5 & 100.0 \\
    \midrule
    45 & Mix-L & 0.08 (0.03) & 19.84 (6.69) & 0.18 (0.16) & 12.5 & 100.0 \\
    & Mix-AL & 0.01 (0.01) & 0.11 (0.10) & 0.12 (0.10) & 3.3 & 100.0 \\
    & Mix-HP-L & 0.09 (0.04) & 14.64 (5.97) & 0.16 (0.14) & 27.2 & 100.0 \\
    & Mix-HP-AL & 0.01 (0.01) & 0.19 (0.19) & 0.12 (0.10) & 2.9 & 100.0 \\
    \midrule
    22.5 & Mix-L & 0.18 (0.08) & 27.68 (11.39) & 0.25 (0.22) & 12.8 & 100.0 \\
    & Mix-AL & 0.04 (0.03) & 0.42 (0.39) & 0.14 (0.12) & 4.3 & 100.0 \\
    & Mix-HP-L & 0.18 (0.07) & 17.76 (8.37) & 0.27 (0.27) & 22.5 & 100.0 \\
    & Mix-HP-AL & 0.04 (0.02) & 0.41 (0.51) & 0.18 (0.17) & 3.1 & 100.0 \\
    \bottomrule
  \end{tabular}
\end{table}
}

\clearpage 
\section{Additional Results on the ADNI Analysis}\label{app:adni}

\begin{figure}[h]
  \begin{center}\includegraphics[width=0.8\textwidth]{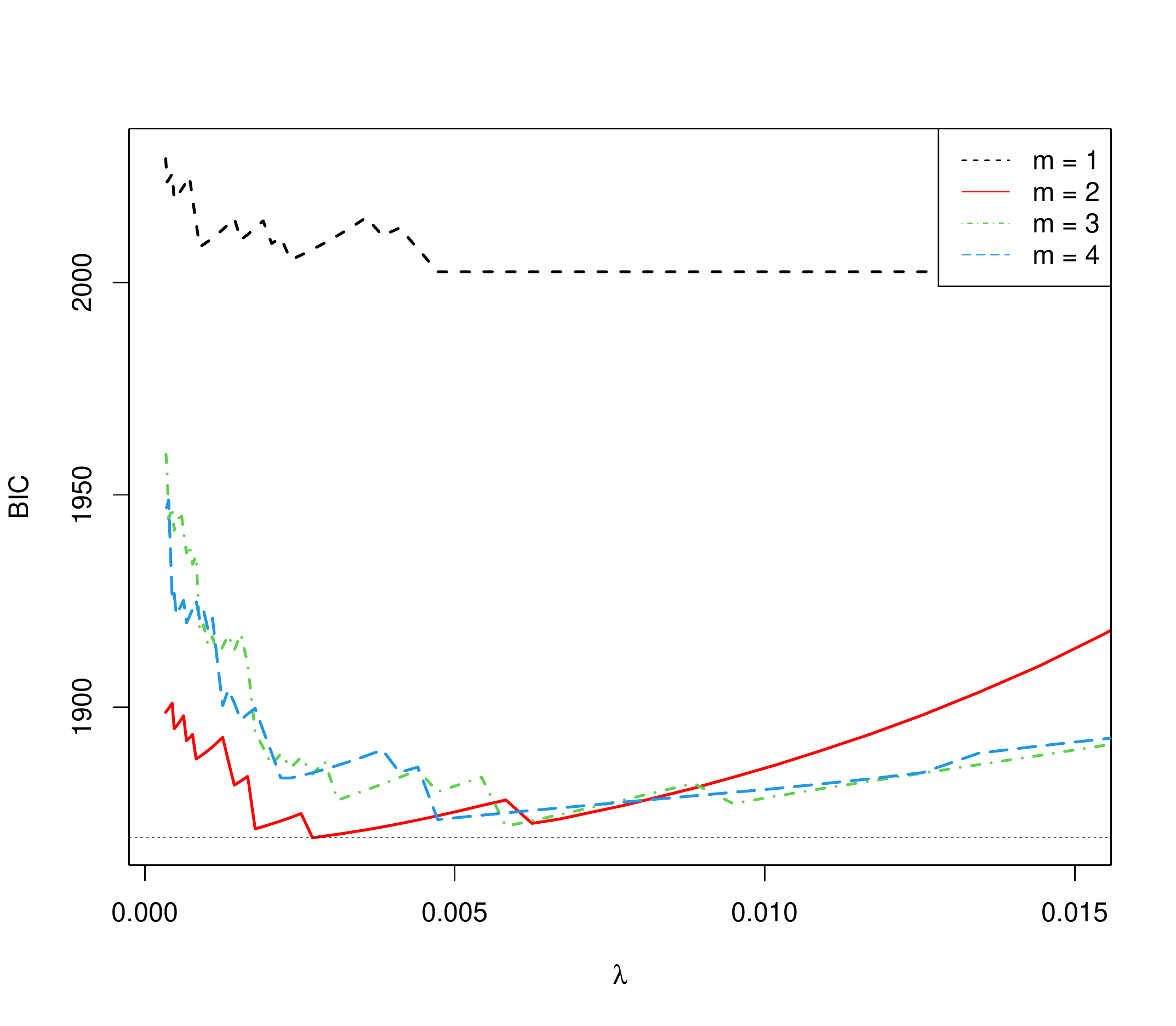}\end{center}
  \caption{{\revcolor{red}ADNI study: selection of tuning parameters including number of components $m\in\{1,2,\ldots, 4\}$ and $\lambda$ for the left ventricles volume data. (To facilitate visualization, the curve for $m=5$ is not shown as its values are much higher.)}
  }\label{fig:adni_bicplot}
\end{figure}

\begin{figure}[h]
  \begin{center}\includegraphics[width=0.6\textwidth]{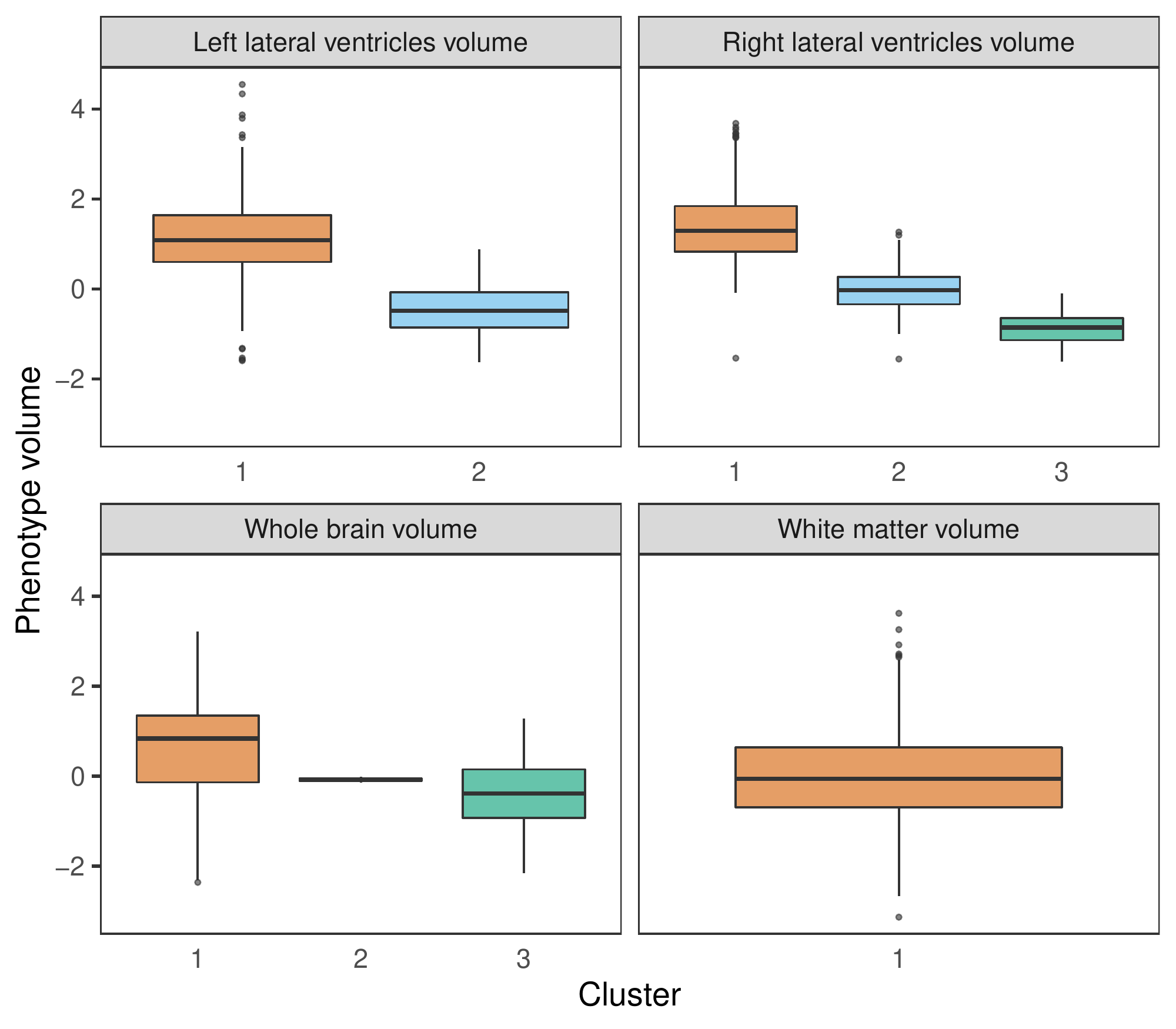}\end{center}
  \caption{ADNI study: boxplots of the four imaging phenotypes across different clusters.
  }\label{fig:adni_boxplot}
\end{figure}

{\def\baselinestretch{1}
\begin{table}[h]
  \centering \caption{ADNI study: coefficient estimates using Mix-HP-AL. The numbers of clusters are 2,3,3,1 for the four imaging phenotypes, respectively. Zero values are shown as blanks. Variables with heterogeneous effects are marked in bold. } \label{tab:res_adni_0} \def\arraystretch{1}
  {\scriptsize
  \begin{tabular}{llrrrrr}
    \toprule
    SNP & Gene & $\widehat{\bphi}_1$ & $\widehat{\bphi}_2$ & $\widehat{\bphi}_3$ \\
    \midrule
    \multicolumn{5}{c}{\it Left Lateral Ventricles Volumes} \\
    rs2025935 & CR1 & 0.06 & 0.06 &  \\
    rs381852 & CDC20B, LOC493869 & 0.20 & 0.20 &  \\
    \textbf{rs1182190} & GNA12 & 0.32 & $-$0.32 &  \\
    \textbf{rs1874445} & MRVI1 & 0.12 & $-$0.12 &  \\
    \textbf{rs7929589} & MS4A4E & 0.06 & $-$0.06 &  \\
    rs12146713 & NUAK1 & 0.16 & 0.16 &  \\
    rs3865444 & CD2AP & $-$0.09 & $-$0.09 &  \\
    \textbf{rs273653} & CD2AP & 0.25 & $-$0.25 &  \\
    $\widehat\sigma$ &  & 0.98 & 0.48 & \\
    $\widehat\pi$ &  & 0.39 & 0.61 & \\
    \midrule
    \multicolumn{5}{c}{\it Right Lateral Ventricles Volumes} \\
    \textbf{rs2025935} & CR1 &  & 0.30 & $-$0.30 \\
    \textbf{rs3737002} & CR1 & 0.22 &  & $-$0.22 \\
    \textbf{rs6709337} & BIN1 & 0.09 &  & $-$0.09 \\
    \textbf{rs798532} & GNA12 &  & 0.15 & $-$0.15 \\
    rs1182197 & GNA12 & $-$0.04 & $-$0.04 & $-$0.04 \\
    \textbf{rs1874445} & MRVI1 & 0.23 &  & $-$0.23 \\
    rs812086 & CSRP3-AS1 & $-$0.27 & $-$0.27 & $-$0.27 \\
    \textbf{rs677909} & PICALM &  & 0.28 & $-$0.28 \\
    rs12146713 & NUAK1 & 0.19 & 0.19 & 0.19 \\
    rs1822381 & MSI2 & $-$0.03 & $-$0.03 & $-$0.03 \\
    rs3865444 & CD2AP & $-$0.22 & $-$0.22 & $-$0.22 \\
    \textbf{rs273653} & CD2AP & 0.11 &  & $-$0.11 \\
    $\widehat\sigma$ &  & 0.85 & 0.42 & 0.29 \\
    $\widehat\pi$ &  & 0.29 & 0.40 & 0.31 \\
    \midrule
    \multicolumn{5}{c}{\it Whole Brain Volumes} \\
    rs10127904 & CR1 & $-$0.13 & $-$0.13 & $-$0.13 \\
    rs16823787 &  & 0.12 & 0.12 & 0.12 \\
    rs2063454 & TFDP2 & 0.22 & 0.22 & 0.22 \\
    rs9473121 & CD2AP & $-$0.10 & $-$0.10 & $-$0.10 \\
    rs10457481 &  & $-$0.24 & $-$0.24 & $-$0.24 \\
    rs854524 & PPP1R9A & 0.13 & 0.13 & 0.13 \\
    \textbf{rs7797990} &  & 0.34 &  & $-$0.34 \\
    rs17745273 & PICALM & $-$0.13 & $-$0.13 & $-$0.13 \\
    rs8756 & HMGA2 & 0.09 & 0.09 & 0.09 \\
    rs1635291 & LINC02210-CRHR1 & $-$0.14 & $-$0.14 & $-$0.14 \\
    $\widehat\sigma$ &  & 0.78 & 0.02 & 0.61 \\
    $\widehat\pi$ &  & 0.52 & 0.01 & 0.47 \\
    \midrule
    \multicolumn{5}{c}{\it White Matter Volumes} \\
    rs3818361 & CR1 & 0.03 &  &  \\
    rs12485574 &  & 0.13 &  &  \\
    rs1867667 & VCAN & $-$0.08 &  &  \\
    rs1385741 & CD2AP & $-$0.01 &  &  \\
    $\widehat\sigma$ &  & 1.08 &  & \\
    \bottomrule
  \end{tabular}
  }
\vspace{-1cm}
\end{table}
}

\clearpage 
\section{Additional Results on the Suicide Risk Analysis}\label{app:suicide}

We perform group-wise principal component analysis (PCA) and use
each leading factor to summarize the information of each variable group/category.

{\def\baselinestretch{1}\normalsize
\begin{table}[h]
  \centering \caption{Suicide risk study: factor loadings from group-wise PCA.}\label{tab:app_pca}
  \begin{tabular}{lrr}
    \toprule
    Variable & Component 1 & Component 2 \\
    \midrule
    \multicolumn{3}{c}{\it Demographic Factor} \\
    Male householder rate  & 0.50 & 0.46 \\
    Household size         & 0.56 & $-$0.39 \\
    \% Population under 18 & 0.50 & $-$0.53 \\
    \% White race          & 0.42 & 0.60 \\
    \% Variation explained & 55.7\% & 29.4\% \\
    \midrule
    \multicolumn{3}{c}{\it Academic Factor} \\
    Average CAPT           & 0.52 & 0.10 \\
    Graduation rate        & 0.53 & 0.28 \\
    Dropout rate           & $-$0.52 & $-$0.34 \\
    Attendance rate        & 0.42 & $-$0.89 \\
    \% Variation explained & 77.0\% & 14.1\% \\
    \midrule
    \multicolumn{3}{c}{\it Behavioral Factor} \\
    \% Serious incidence   & 0.58 & 0.40 \\
    Incidence rate         & 0.57 & $-$0.82 \\
    Serious incidence rate & 0.58 & 0.40 \\
    \% Variation explained & 96.20\% & 3.70\% \\
    \midrule
    \multicolumn{3}{c}{\it Economic Factor} \\
    Median income & 0.70 & 0.71 \\
    Free lunch rate & $-$0.71 & 0.70 \\
    \% Variation explained & 84.80\% & 15.20\% \\
    \bottomrule
  \end{tabular}
\end{table}
}

\vspace{-2cm}

\clearpage
\section{Salary and Performance in Major League Baseball}\label{app:baseball}

The data contains salaries for major league baseball players for the year 1992, along with their performance and status measures from the year 1991. These players played at least one game in the 1991 and 1992 seasons, and pitchers were not included. The data is available on the website of {\it Journal of Statistics Education} ({\it \url{www.amstat.org/ publications/jse}}). The main interest is to examine which performance measures and status indicators play important roles in determining the salary of a player. 
The analysis was first conducted using regularized mixture regression by \citet{khalili2012variable} and using Bayesian variable selection by \citet{lee2016bayesian}.

There are 12 numerical performance measures including batting average ($x_1$), on-base percentage ($x_2$), runs ($x_3$), hits ($x_4$), doubles ($x_5$), triples ($x_6$), home runs ($x_7$),
runs batted in ($x_8$), walks ($x_9$), strikeouts ($x_{10}$), stolen
bases ($x_{11}$), and errors ($x_{12}$); and there are 4 indicator variables including free agency eligibility ($x_{13}$), free agent in 1991/1992 ($x_{14}$), arbitration
eligibility ($x_{15}$), and arbitration in 1991/2 ($x_{16}$), which measures how free each player was to move to another team. Previous work suggested that there could be interaction effects between $x_1, x_3, x_7, x_8$ and $x_{13}$--$x_{16}$. This leads to in total $p=32$ candidate predictors for the analysis. Following \citet[Section 3, List 10]{watnik1998}, we standardize $x_1$--$x_{12}$ before introducing the interaction terms, and some outliers are removed to result in a sample size of $n=331$. The log-transformed salary is used as the response.

\citet{khalili2012variable} used a two-component Gaussian mixture regression with equal variance. \citet{lee2016bayesian} then suggested using unequal variances. We thus compare four models, the two-component Mix-AL and the two-component Mix-HP-AL, either with or without the assumption of equal variance. The BIC values, and the out-of-sample predictive log-likelihood values are presented in Table~\ref{tab:app_base_ev}. The out-of-sample predictive log-likelihood values (with standard error in parenthesis) are obtained from a random-splitting procedure. Each time the data is split to 80\% training data for model fitting and 20\% testing data for out-of-sample evaluation, and the procedure is repeated 500 times and the results are averaged. Based on the out-of-sample predictive likelihood, our proposed method Mix-HP-AL with unequal variances performs the best.

{\def\baselinestretch{1}\normalsize
\begin{table}[htp]
  \centering \caption{Baseball salary study: compare models with/without heterogeneity pursuit and equal variance assumption.}\label{tab:app_base_ev} \def\arraystretch{1}
  \begin{tabular}{lrrrrrr}
    \toprule
    & \multicolumn{2}{c}{Mix-AL} & \multicolumn{2}{c}{Mix-HP-AL} \\
    \midrule
    & DV & EV & DV & EV\\
    \midrule
    BIC & 459.77 & 438.16 & 438.64 & 456.73 \\
    Pred & -55.2 (4.48) & -56.9 (6.03) & -51.7 (4.73) & -56.8 (7.07) \\
    \bottomrule
  \end{tabular}
\end{table}
}

The results of using Mix-HP-AL with unequal variances are presented in Table~\ref{tab:res_baseball_0}. The mixing probabilities are estimated as 0.71 and 0.29, and the two fd variances are quite close to each other. A total of 16 predictor terms are selected; interestingly, among the selected terms, only three terms, namely, the free agent eligibility, the arbitration eligibility and the interaction between runs and arbitration eligibility, are identified to be sources of heterogeneity according to Definition \ref{def:2}. This finding coincides with the conclusion in \citet{lee2016bayesian} that the eligibility of free agent or free arbitration is the key factor to distinguish the two clusters. Our analysis suggests that many performance measures of the players, on the other hand, generally relate to the salary level in a common and homogeneous way. We notice that being a free agent or having arbitration in the 1991/2 seasons is not associated with a higher salary, conditioning on the other selected terms. Intriguingly, this may be explained by the Player's Union argument that owners colluded to keep the salary of free agents lower in the 1991/2 seasons \citep{watnik1998}.

{\def\baselinestretch{1}
\begin{table}[tbp]
  \centering \caption{Baseball salary study: the scaled coefficient
    estimates using Mix-HP-AL. Zero values are shown as blanks. 
    Variables with heterogeneous effects are marked in bold.} \label{tab:res_baseball_0}
  \def\arraystretch{0.95} {\small
  \begin{tabular}{lrrrrrr}
    \toprule
    & $\widehat{\bphi}_1$ & $\widehat{\bphi}_2$\\
    \midrule
  Intercept & 23.93 & 23.93 \\ 
  $x_{1}$ &  &  \\ 
  $x_{2}$ &  &  \\ 
  $x_{3}$ &  &  \\ 
  $x_{4}$ & 1.30 & 1.30 \\ 
  $x_{5}$ &  &  \\ 
  $x_{6}$ & $-$0.23 & $-$0.23 \\ 
  $x_{7}$ &  &  \\ 
  $x_{8}$ & 0.48 & 0.48 \\ 
  $x_{9}$ & 0.39 & 0.39 \\ 
  $x_{10}$ & $-$0.40 & $-$0.40 \\ 
  $x_{11}$ & 0.25 & 0.25 \\ 
  $x_{12}$ &  &  \\ 
  $\bm{x_{13}}$ & 9.77 & 1.71 \\ 
  $x_{14}$ & $-$0.82 & $-$0.82 \\ 
  $\bm{x_{15}}$ & 6.57 & 4.30 \\ 
  $x_{16}$ & $-$2.02 & $-$2.02 \\ 
  $x_{1}*x_{13}$ & $-$0.52 & $-$0.52 \\ 
  $x_{1}*x_{14}$ &  &  \\ 
  $x_{1}*x_{15}$ &  &  \\ 
  $x_{1}*x_{16}$ &  &  \\ 
  $x_{3}*x_{13}$ &  &  \\ 
  $x_{3}*x_{14}$ & 0.70 & 0.70 \\ 
  $\bm{x_{3}*x_{15}}$ & 0.50 & $-$0.50 \\ 
  $x_{3}*x_{16}$ &  &  \\ 
  $x_{7}*x_{13}$ & 0.25 & 0.25 \\ 
  $x_{7}*x_{14}$ &  &  \\ 
  $x_{7}*x_{15}$ &  &  \\ 
  $x_{7}*x_{16}$ &  &  \\ 
  $x_{8}*x_{13}$ &  &  \\ 
  $x_{8}*x_{14}$ & 0.86 & 0.86 \\ 
  $x_{8}*x_{15}$ &  &  \\ 
  $x_{8}*x_{16}$ & 1.83 & 1.83 \\
  $\widehat\sigma$ & 0.23 & 0.27 \\
  $\widehat\pi$ & 0.71 & 0.29 \\
    \bottomrule
  \end{tabular}
  }
\end{table}
}

It is also interesting to examine the estimated cluster pattern. As
seen from Figure~\ref{fig:hist_baseball}, the majority of players are in cluster 1, in which the players tend to have either higher or lower salaries than the average, and cluster 2 mainly consists of the players with average or ``normal'' level of salaries. The players with the eligibility of being free agent or having arbitration are mostly those with much higher salaries in cluster 1, but not in cluster 2. This explains why these variables have much larger effects in cluster 1 than in cluster 2. 

\begin{figure}[htp]
  \includegraphics[width=\textwidth]{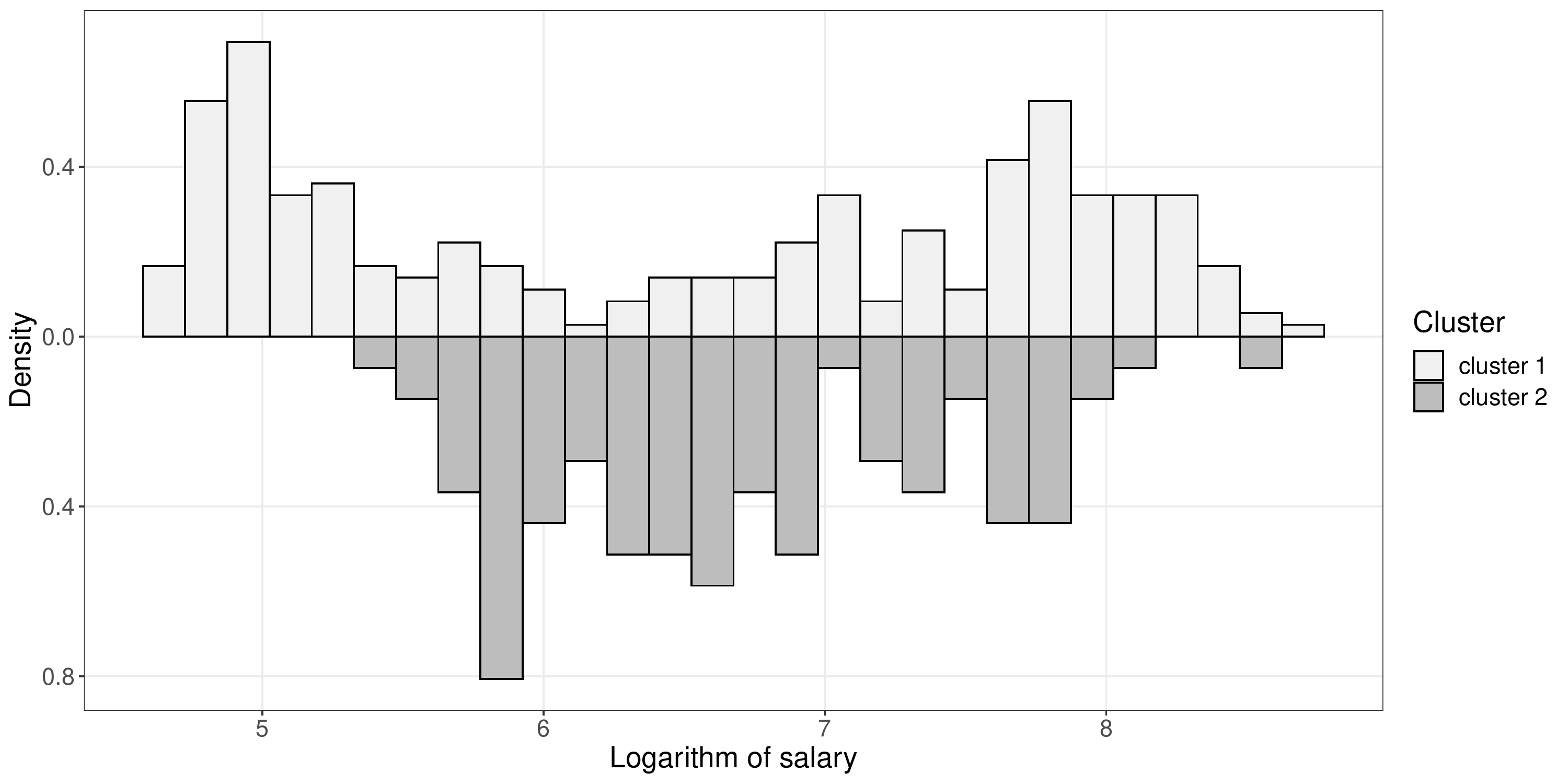}
  \caption{Baseball salary study: the distribution of $\log (\mbox{salary})$ for the two
    clusters.}\label{fig:hist_baseball}
\end{figure}


\end{appendices}

\clearpage

\vspace*{-8pt}

\vskip 0.2in

\bibliographystyle{chicago}
\bibliography{ref,reference}

\end{document}